%% file: CSMINRES25.tex
\documentclass[final]{siamltex}

\usepackage[notref,notcite]{showkeys}
\usepackage{verbatim}
\usepackage{amsfonts}
\usepackage{amsmath,amssymb}
\usepackage{graphicx}
\usepackage{epstopdf}
\usepackage{enumerate,url}
\usepackage{xspace}    % "OSCEF\xspace is" produces "OSCEF is"
\usepackage[vlined,ruled,nofillcomment,linesnumbered]{algorithm2e}
\SetAlgoSkip{}\SetAlgoInsideSkip{smallskip}

\newenvironment{algo}[1]
{\begin{algorithm}[#1]%
    \small%
    \DontPrintSemicolon%
    \SetArgSty{texttsf}%
    \SetTitleSty{textsf}{}%
    \SetNlSty{textrm}{}{}%
    \SetKwInput{Inputs}{input}%
    \SetKwInput{Outputs}{output}%
    \SetKwData{Converged}{converged}%
    \SetKwComment{tcc}{[}{]}
}
{\end{algorithm}}

\input{CsminresMacros.tex}

\newtheorem{example}{Example}[section]
\newtheorem{remark}{Remark}[section]

\newcommand{\e}[1]{\text{e}{#1}}
\newcommand{\myhalf}{\frac12}
\newcommand{\mystrut}{\rule[-1.9ex]{0pt}{5.3ex}}
\newcommand{\tablestrut}{\rule[-1ex]{0pt}{3.5ex}}
\newcommand{\underTj}{\underline{T_j}}
\newcommand{\underTk}{\underline{T_k}}
\newcommand{\underTkp}{\underline{T_{k+1}}}

\newcommand{\red}  [1]{\textcolor{black}      {\null #1}} % red or black in brace

\usepackage[pdftex,
  pdftitle={Minimal Residual Methods for Complex Symmetric, Skew Symmetric, and Skew Hermitian Systems},
  pdfauthor={Sou-Cheng (Terrya) Choi},
  pdfsubject={Minimal Residual Methods for Complex Symmetric, Skew Symmetric, and Skew Hermitian Systems},
  pdfkeywords={Krylov subspace method, Lanczos process,
    conjugate-gradient method, minimum-residual method, ill-posed
    problem, singular least-squares problem, sparse matrix, MINRES, 
    MINRES-QLP, CS-MINRES, CS-MINRES-QLP, SS-MINRES, SS-MINRES-QLP,
  SH-MINRES, SH-MINRES-QLP},
  pdfpagemode=UseOutlines,
  pdfstartview=FitH,
  bookmarks,
  bookmarksopen,
  colorlinks,
  linkcolor=blue,
  citecolor=blue,
  urlcolor=red,
  hypertexnames=false
]{hyperref}

\title{Minimal
Residual Methods for Complex Symmetric, Skew Symmetric, and Skew
Hermitian Systems\thanks{Draft CSMINRES25.tex of January 6, 2014. %\today.
Argonne Tech. Rep. ANL/MCS-P3028-0812.
}
}

\author{Sou-Cheng T. Choi\thanks{Computation Institute, University of
    Chicago, Chicago, IL 60637 (sctchoi@uchicago.edu). This work was
    supported in part by the \red{U.S.\ Dept.~of Energy, Office of Science, Advanced
      Scientific Computing Research} under
    Contract DE-AC02-06CH11357 and by the National Science Foundation
    grant SES-0951576 through the Center for Robust Decision Making on
    Climate and Energy Policy.}}

\begin{document}

\maketitle

\begin{center}
\emph{Dedicated to Michael Saunders's 70th birthday}
\smallskip
\end{center}

%%%%%%%%%%%%%%%%%%%%%%%%%%%%%%%%%%%%%%%%%%%%%%%%%%%%%%%%%%
\begin{abstract}
%%%%%%%%%%%%%%%%%%%%%%%%%%%%%%%%%%%%%%%%%%%%%%%%%%%%%%%%%%
While there is no lack of efficient Krylov subspace solvers for
Hermitian systems, \red{few exist} for complex symmetric, skew
symmetric, or skew Hermitian systems, which are increasingly important
in modern applications including quantum dynamics, electromagnetics,
and power systems.
For a large\red{,} consistent\red{,} complex symmetric system, one may
apply a non-Hermitian Krylov subspace method disregarding the symmetry
of $A$, or a Hermitian Krylov solver on the equivalent normal equation
or an augmented system twice the original dimension. These have the
disadvantages of increasing memory, conditioning, or computational
costs.  An exception is a special version of QMR by Freund (1992), but
that may be affected by \red{nonbenign} breakdowns unless look-ahead
is implemented; furthermore, it is designed for only consistent and
nonsingular problems.  Greif and Varah (2009) adapted CG for
nonsingular skew symmetric linear systems that are necessarily and
restrictively of even order.

We extend the symmetric and Hermitian algorithms MINRES and MINRES-QLP
by Choi, Paige\red{,} and Saunders (2011) to complex symmetric, skew
symmetric, and skew Hermitian systems. In particular, MINRES-QLP uses
a rank-revealing QLP decomposition of the tridiagonal matrix from a
three-term recurrent complex symmetric Lanczos process. Whether the
systems are real or complex, singular or invertible, compatible or
inconsistent, MINRES-QLP computes the unique minimum-length
\red{(i.e., pseudoinverse)} solutions.  It is a significant extension
of MINRES by Paige and Saunders (1975) with enhanced stability and
capability.

%We also derive preconditioned CS-MINRES, stopping rules, and
%efficient estimates of the solution and residual norms, matrix norm,
%and condition number. Like MINRES-QLP, CS-MINRES has elegant
%theoretical properties, and is sufficiently robust for ill-cond
%itioned and ill-posed problems. We present extensive numerical
%experiments to demonstrate its scalability and robustness.
\begin{comment}
Noting the similarities in the definitions of skew symmetric matrices
$(A=-A\T \in \mathbb{R}^{n\times n})$ and complex symmetric matrices
$(A=A^T \in \mathbb{C}^{n\times n})$, we further extend MINRES-QLP for
skew symmetric equations.  Lastly, via a simple transformation of skew
Hermitian systems $(A=-A^* \in \mathbb{C}^{n\times n})$, we show that
MINRES-QLP can readily solve this class of problems.
\end{comment}
%%%%%%%%%%%%%%%%%%%%%%%%%%%%%%%%%%%%%%%%%%%%%%%%%%%%%%%%%%
\end{abstract}
%%%%%%%%%%%%%%%%%%%%%%%%%%%%%%%%%%%%%%%%%%%%%%%%%%%%%%%%%%

%%%%%%%%%%%%%%%%%%%%%%%%%%%%%%%%%%%%%%%%%%%%%%%%%%%%%%%%%%
%%%%%%%%%%%%%%%%%%%%%%%%%%%%%%%%%%%%%%%%%%%%%%%%%%%%%%%%%%


\begin{DOI}
   xxx/xxxxxxxxx
\end{DOI}

%%%%%%%%%%%%%%%%%%%%%%%%%%%%%%%%%%%%%%%%%%%%%%%%%%%%%%%%%%
\section{Introduction}  \label{sec:intro}
%%%%%%%%%%%%%%%%%%%%%%%%%%%%%%%%%%%%%%%%%%%%%%%%%%%%%%%%%%

Krylov subspace methods \red{for linear systems} are generally divided
into two classes: \red{those} for Hermitian matrices
(e.g.\red{,}\ \CG~\cite{HS52}, \MINRES~\cite{PS75},
\SYMMLQ~\cite{PS75}, \MINRES-\QLP~\cite{CPS11,CS12,CS12b,C06}) and
those for general matrices without such symmetries
(e.g.\red{,}\ \BCG~\cite{F76}, \GMRES~\cite{SS86}, \QMR~\cite{FN91},
\BICGSTAB~\cite{V92}, \LSQR~\cite{PS82a,PS82b}, \red{and
  IDR$(s)$~\cite{SV08})}.  Such a division is largely due to
historical reasons in numerical linear algebra---the most prevalent
structure for matrices arising from practical applications being
Hermitian (which reduces to symmetric for real
matrices). However\red{,} other types of symmetry structures, notably
complex symmetric, skew symmetric, and skew Hermitian matrices, are
becoming increasingly common in modern applications. Currently,
\red{except} possibly for storage and matrix-vector products, these
are treated \red{as} general matrices with no symmetry structures. The
algorithms in this article go substantially further in developing
specialized Krylov subspace algorithms designed at the outset to
exploit \red{the} symmetry structures. In addition, our algorithms
constructively reveal the (numerical) compatibility and singularity of
a given linear system\red{;} users do not have to know these
properties a priori.

%TODO \MINRESQLP~\cite{C06,CPS11} is a Krylov subspace method for
%computing the unique minimum-length and minimum-residual solution
%(also known as the pseudoinverse solution) $x$ to a nonsingular
%linear system~(\ref{eqn1a}), a singular linear system~(\ref{eqn1b}),
%or a necessarily singular least-squares (LS) problem~(\ref{eqn1c}):
We are concerned with iterative methods for solving a large linear system
$Ax=b$ or the more general minimum-length least-squares (\LS) problem
\begin{equation}  \label{eqn4b}
  \min \norm{x}_2 \quad \text{s.t.} \quad x \in 
  \arg\min_{x \in \mathbb{C}^{n}}\norm{Ax-b}_2,
\end{equation}
where $A \in \mathbb{C}^{n \times n}$ is complex symmetric ($A=A\T$\;)
or skew Hermitian ($A=-A^*$), and possibly singular\red{,} and $b \in
\mathbb{C}^{n}$.  Our results are directly applicable to problems with
symmetric or skew symmetric matrices $A=\pm A\T \in \mathbb{R}^{n
  \times n}$ and real vectors $b$.  $A$ may exist only as an operator
for returning the product $Ax$.

The solution of~\eqref{eqn4b}, called the \textit{minimum-length} or
\textit{pseudoinverse} solution~\cite{GV12}, is formally given by
$x^\dagger = (A^* A)^\dagger A^* b$,
% = (A^2)^\dagger Ab = (A^\dagger)^2 Ab$, 
where $A^\dagger$ denotes the pseudoinverse of $A$. The pseudoinverse
is continuous under perturbations $E$ for which
$\rank{(A+E)}=\rank{(A)}$~\cite{S69}, and $x^\dagger$ is continuous
under the same condition.  Problem~\eqref{eqn4b} is then
well-posed~\cite{Had1902}.

%NOTE: Do not have any space before "\footnote" below.
Let $A=U \Sigma U\T $ be a Takagi decomposition~\cite{HJ}, a
singular-value decomposition (\SVD) specialized for a complex
symmetric matrix, with $U$ unitary ($U^*U=I$) and $\Sigma \equiv
\diag(\smat{\sigma_1,\ldots,\sigma_n})$ real non-negative and
$\sigma_1 \ge \sigma_2 \ge \cdots \ge \sigma_r > 0$, where $r$ is the
rank of $A$.  We define the condition number of $A$ to be $\kappa(A) =
\smash[b]{\frac{\sigma_1}{\sigma_r}}$, and we say that $A$ is
ill-conditioned if $\kappa(A)\gg 1$.  Hence a mathematically
nonsingular matrix (e.g., $A=\smat{1 & 0 \\ 0 & \varepsilon}$, where
$\varepsilon$ is the machine precision) could be regarded as
numerically singular. Also, a singular matrix could be
well-conditioned or ill-conditioned. For a skew Hermitian matrix, we
use its (full) eigenvalue decomposition $A=V\Lambda V^*$, where
$\Lambda$ is a diagonal matrix of imaginary numbers (possibly zeros;
in conjugate pairs if $A$ is real, i.e., skew symmetric) and $V$ is
unitary.\footnote{Skew Hermitian (symmetric) matrices are, like
  Hermitian matrices, unitarily diagonalizable (i.e.,
  normal~\cite[Theorem~24.8]{TB}).} We define its condition number as
$\kappa(A) = \smash[b]{\frac{|\lambda_1|}{|\lambda_r|}}$, the ratio of
the largest and smallest nonzero eigenvalues in magnitude.

\bigskip
\begin{example} We contrast the five classes of symmetric 
or Hermitian matrices by their definitions and small instances of
order $n=2$:
\begin{align*} 
  \mathbb{R}^{n\times n} \ni A &=A\T =\bmat{1 & 5 \\ 5 & 1} \text{ is
    symmetric}.  
  \\ \mathbb{C}^{n\times n} \ni A &=A^* =\bmat{1 & 1-2i
    \\ 1+2i & 1} \text{ is Hermitian (with real diagonal)}.
  \\ \mathbb{C}^{n\times n} \ni A &=A\T =\bmat{2+i & 1-2i \\ 1-2i & i}
  \text{ is complex symmetric (with complex
    diagonal)}.  
  \\ \mathbb{R}^{n\times n} \ni A
  &=-A\T =\bmat{0 & 5 \\ -5 & 0} \text{ is skew symmetric (with zero
    diagonal)}.  
  \\ \mathbb{C}^{n\times n} \ni A &=-A^* =\bmat{0 &
    1-2i \\ -1-2i & \red{i}} \text{ is skew Hermitian (with \red{imaginary} diagonal)}.
\end{align*}
\end{example}

%The conjugate gradient (\CG) method~\cite{HS52}, \SYMMLQ and
%\MINRES~\cite{PS75} are Krylov subspace methods 
\CG, \SYMMLQ, and \MINRES are designed for solving nonsingular
symmetric systems $Ax = b$.  \CG is efficient on symmetric positive
definite systems. For indefinite problems, \SYMMLQ and \MINRES are
reliable even if $A$ is ill-conditioned.

Choi~\cite{C06} \red{appears} to be the first \red{to} comparatively
\red{analyze} the algorithms on singular symmetric and Hermitian
problems.  On (singular) incompatible problems \CG and \SYMMLQ
iterates $x_k$ diverge to some nullvectors of
$A$~\cite[Propositions~2.7, 2.8, and 2.15; Lemma 2.17]{C06}.  \MINRES
often seems more desirable to users because its residual norms are
monotonically decreasing.  On singular compatible systems, \MINRES
returns $x^\dagger$~\cite[Theorem~2.25]{C06}.  On singular
incompatible systems, \MINRES remains reliable if it is terminated
with a suitable stopping rule that monitors
$\norm{Ar_k}$~\cite[Lemma~3.3]{CPS11},
%SC: moved up the phrase about stopping rule
%\CSMINRES should also be reliable on singular incompatible systems
%if $A$ is not too ill-conditioned.
but the solution is generally not
$x^\dagger$~\cite[Theorem~3.2]{CPS11}. \MINRESQLP
\cite{CPS11,CS12,CS12b,C06} is a significant extension of \MINRES,
capable of computing $x^\dagger$, simultaneously minimizing residual
and solution norms.  The additional cost of \MINRESQLP is moderate
relative to \MINRES: $1$ vector in memory, $4$ axpy \red{operations}
($y \leftarrow \alpha x + y$), and $3$ vector \red{scalings} ($x
\leftarrow \alpha x$) per iteration.  The efficiency of \MINRES is
partially, and in some cases almost fully, retained in \MINRESQLP by
transferring from a \emph{\MINRES phase} to a \emph{\MINRESQLP phase}
only when an estimated $\kappa(A)$ exceeds a user-specified value.
The \MINRES phase is optional, consisting of only \MINRES iterations
for nonsingular and well-conditioned subproblems.  The \MINRESQLP
phase handles less well-conditioned and possibly numerically singular
subproblems.  In all iterations, \MINRESQLP uses \QR factors of the
tridiagonal matrix from a Lanczos process and then applies a second
\QR decomposition on the conjugate transpose of the upper-triangular
factor to obtain and reveal the rank of a lower-tridiagonal form.
%Thus \CSMINRES may be regarded as an extension of \MINRESQLP and does
%not suffer from non-benign breakdowns.
On nonsingular systems, \MINRESQLP enhances the accuracy (with
\red{smaller} rounding errors) and stability of \MINRES.  It is
applicable to symmetric and Hermitian problems with no traditional
restrictions such as nonsingularity and definiteness of~$A$ or
compatibility of~$b$.

%TODO cite Helmhotz, Hankel, all iterative  methods
The aforementioned established Hermitian methods are not\red{,
  however,} directly applicable to complex or skew symmetric
equations.  For consistent complex symmetric problems, which could
arise in Helmholtz equations, linear systems that involve Hankel
matrices, or applications in quantum dynamics, electromagnetics, and
power systems, we may apply a non-Hermitian Krylov subspace method
disregarding the symmetry of $A$ or a Hermitian Krylov solver (such as
\CG, \SYMMLQ, \MINRES, or \MINRESQLP) on the equivalent normal
equation or an augmented system twice the original dimension. They
suffer increasing memory, conditioning, or computational costs.  An
exception\footnote{It is noteworthy that among direct methods for
  large sparse systems, MA57 and ME57~\cite{D09} are available for
  real\red{, Hermitian}, and complex symmetric problems.}
is a special version of \QMR by Freund (1992)~\cite{F92}, which takes
advantage of the matrix symmetry by using an unsymmetric Lanczos
framework.  Unfortunately, the algorithm may be affected by
\red{nonbenign} breakdowns unless a look-ahead strategy is
implemented.  Another less than elegant feature of \QMR is \red{that}
the vector norm of choice is induced by the inner product $x\T y$ but
it is not a proper vector norm (e.g., $0 \neq x\T :=\smat{ 1 & i }$,
where $i=\sqrt{-1}$, yet $x\T x = 0$).  Besides, \QMR is designed for
only nonsingular and consistent problems.  Inconsistent complex
symmetric problems~\eqref{eqn4b} could arise from shifted problems in
inverse or Rayleigh quotient iterations; mathematically or numerically
singular or inconsistent systems, in which $A$ or $b$ are vulnerable
to errors due to measurement, discretization, truncation, or
round-off. In fact, \QMR and most non-Hermitian Krylov solvers (other
than LSQR) fail to converge to $x^\dagger$ on an
example as simple as $A = i \; \diag(\smat{1 \\ 0})$ and $b = i
\smat{1 \\ 1}$, for which $x^\dagger =\smat{1 \\ 0}$.

Here we extend the symmetric and Hermitian algorithms \MINRES and
\MINRESQLP%~\cite{CPS11, CS12, CS12b} to complex symmetric systems.
The main aim is to deal reliably with compatible or incompatible
systems and to return the \emph{unique} solution
of~\eqref{eqn4b}. Like \QMR and the Hermitian Krylov solvers, \red{our
  approach} exploits the matrix symmetry.
%We use ``\CSMINRES'' to refer to the \emph{mega} algorithm
%(with phase transfer) or the \emph{standalone} (without a \CSMINRESQLP
%phase), but the usage should be clear from the context \CSMINRES
%resembles \MINRESQLP; 

Noting the similarities in the definitions of skew symmetric matrices
($A=-A\T \in \mathbb{R}^{n\times n}$\red{)} and complex symmetric matrices
and motivated by algebraic Riccati equations~\cite{I84} and more
recent, novel applications of Hodge theory in \red{data
  mining}~\cite{JLYY11,GL11}, we evolve \MINRESQLP \red{further} for
solving skew symmetric linear systems. Greif and Varah~\cite{GV09}
adapted CG for nonsingular skew symmetric linear systems that are
\emph{skew-$A$ conjugate}, meaning $-A^2$ is symmetric positive
definite. The algorithm is further restricted to $A$ of \red{even
  order because} a skew symmetric matrix of odd order is singular.
Our \MINRESQLP extension has no such limitations and \red{is}
applicable to singular problems. For skew Hermitian \red{systems} with
skew Hermitian matrices or operators ($A=-A^* \in \mathbb{C}^{n\times
  n}$), our approach is to transform them into Hermitian systems so
that they \red{can} immediately take advantage of the original
Hermitian version of \MINRESQLP.

%------------------------------------------------ 
\subsection{Notation}  \label{sec:notation}
%------------------------------------------------ 

For an incompatible system, $Ax\approx b$ is shorthand for the \LS
problem~(\ref{eqn4b}). We use ``$\simeq$'' to mean ``approximately
equal to\red{.''}  The letters $i$, $j$, $k$ in subscripts or
superscripts denote integer indices\red{; $i$ may} also represent
$\sqrt{-1}$\red{. We use} $c$ and $s$ \red{for} cosine and sine of
some angle $\theta$; $e_k$ \red{is} the $k$th unit vector; $e$
\red{is} a vector of all ones; and other lower-case letters such as
$b$, $u$, and $x$ (possibly with integer subscripts) denote
\textit{column} vectors.  Upper-case letters $A$, $T_k$, $V_k$,
\dots{} denote matrices, and $I_k$ is the identity matrix of order
$k$.  Lower-case Greek letters denote scalars; in particular,
$\varepsilon \simeq 10^{-16}$ denotes floating-point double precision.
If a quantity $\delta_k$ is modified one or more times, we denote its
values by $\delta_k$, $\delta_k^{(2)}$, \red{and so on}. We use
$\diag(v)$ to denote a diagonal matrix with elements of a vector $v$
on the diagonal. The transpose, conjugate, and conjugate transpose of
a matrix $A$ \red{are} denoted \red{by} $A\T$, $\conj{A}$, and
$A^*=\conj{A}\T$\red{,} respectively.  The symbol $\norm{\,\cdot\,}$
denotes the $2$-norm of a vector ($\norm{x}=\sqrt{x^*x}$) or a matrix
($\norm{A}=\sigma_1$ from $A$'s \SVD). 

%$\mathcal{R}(\cdot)$ the range and $\mathcal{N}(\cdot)$ the nullspace
%of a matrix, while $\mathcal{K}_k(A,b)$ is the $k$th Krylov subspace
%of $A$ and $b$.  SC: made reference to (1.1)

%SC: defined positive definite matrix A symmetric matrix $A$ is
%positive definite if all its eigenvalues are positive and we write $A
%\succ 0$.

%------------------------------------------------ 
\subsection{Overview}
%------------------------------------------------ 

In Section~\ref{sec:review} we briefly review the Lanczos processes
and QLP decomposition before developing the algorithms in Sections
\ref{sec:csminres-standalone}-\ref{sec:transfer}.  Preconditioned
algorithms are described in Section~\ref{sec:pcsminres}.  Numerical
experiments are described in Section~\ref{sec:numerical}. We conclude
with future work and related software in
Section~\ref{sec:conclusions}.  Our pseudocode and a summary of norm
estimates and stopping conditions are given in
Appendices~\ref{sec:pseudo} and~\ref{sec:QLPstop}.

%%%%%%%%%%%%%%%%%%%%%%%%%%%%%%%%%%%%%%%%%%%%%%%%%%%%%%%%%%
\section{Review}  \label{sec:review} 
%%%%%%%%%%%%%%%%%%%%%%%%%%%%%%%%%%%%%%%%%%%%%%%%%%%%%%%%%%

In the following few subsections, we summarize algebraic methods
necessary for our algorithmic development.

%------------------------------------------------ 
\subsection{\red{Saunders} and Lanczos processes}  \label{sec:Lanczos}
%------------------------------------------------ 

%SC: removed "reliable" and added "in terms of numerical stability" to
%Chris's form
Given a complex symmetric operator $A$ and a vector $b$, a
Lanczos-\emph{like}\footnote{We distinguish our process from the
  complex symmetric Lanczos process~\cite{L56_88} used in
  QMR~\cite{F92}.}  process~\cite{BS99}, which we name the
\emph{Saunders process}, computes vectors $v_k$ and tridiagonal
matrices $\underline{T_k}$ according to $v_0 \equiv 0$, $\beta_1
v_1=b$, and then\footnote{Numerically, $p_k = A\conj{v}_k-\beta_k
  v_{k-1}$, $\alpha_k = v_k^* p_k$, $\beta_{k+1}v_{k+1} =
  p_k-\alpha_kv_k$ is slightly better~\cite{P76}.}
\begin{equation} \label{eq:savk}
    p_k = A\conj{v}_k, \qquad \alpha_k = v_k^* p_k, \qquad
    \beta_{k+1}v_{k+1} = p_k-\alpha_kv_k-\beta_kv_{k-1}
\end{equation}
for $k=1,2,\dots,\ell$, where we choose $\beta_k > 0$ to give
$\norm{v_k}=1$.  In matrix form,
\begin{equation}  \label{eq:avk}
  A \conj{V}\!_k \!=\! V_{k+1}\underTk,\quad
  \underTk \!\equiv\! \mbox{\footnotesize
    $\bmat{\alpha_{1} & \beta_{2}
        \\ \beta_{2}  & \alpha_{2} & \ddots
        \\            & \ddots & \ddots & \beta_k
        \\            &  & \beta_k & \alpha_k
        \\            &  &  & \beta_{k+1}}
    $}
    \!\equiv\!
    \bmat{T_k \\ \beta_{k+1}e_k^T},
    \quad V_k \!\equiv\! \bmat{v_1 & \!\cdots\! & v_k}.
\end{equation}
In exact arithmetic, the columns of $V_k$ are orthogonal\red{,} and the
process stops with $k = \ell$ and $\beta_{\ell+1}=0$ for some $\ell
\le n$, and then $A \conj{V}\!_\ell = V_\ell T_\ell$.  For derivation
purposes we assume that this happens, though in practice it is rare
unless $V_k$ is reorthogonalized for each $k$.
%In practice it is likely that $\beta_k > 0$ for all $k$, but
In any case, \eqref{eq:avk} holds to machine precision\red{,} and the
computed vectors satisfy $\norm{V_k}_1 \simeq 1$ (even if $k \gg n$).

If instead we are given a skew symmetric $A$, the following is a
Lanczos process~\cite[Algorithm~1]{GV09}\footnote{Another Lanczos
  process for skew symmetric~$A$ \red{using} a different measure to
  normalize $\beta_{k+1}$ was developed in~\cite{W78,SW93}.}  that
transforms $A$ to a series of expanding, skew symmetric tridiagonal
matrices $T_k$ and generates a set of orthogonal vectors in $V_k$ in
exact arithmetic:
\begin{equation}\label{eq:savk-ss}
    p_k = A v_k,  \qquad - \beta_{k+1}v_{k+1} = p_k -\beta_kv_{k-1}\red{,}
\end{equation}
\red{where $\beta_k > 0$ for $k < \ell$.} Its associated matrix form is 
\begin{equation}  \label{eq:avk-sh}
  A V_k \!=\! V_{k+1}\underTk,\quad
  \underTk \!\equiv\! \mbox{\footnotesize
    $\bmat{0          & \beta_{2}
        \\ - {\beta}_{2} & 0      & \ddots
        \\            & \ddots & \ddots & \beta_k
        \\            &  & - {\beta}_k & 0
        \\            &  &  & - {\beta}_{k+1}}
    $}
    \!\equiv\!
    \bmat{T_k \\  -\beta _{k+1}e_k^T}.
   % \quad V_k \!\equiv\! \bmat{v_1 & \!\cdots\! & v_k}.
\end{equation}

%\begin{assumption}[Orthogonality of $V_k$] \label{assume:orthog} For
%    example, in deriving estimates of certain norms, if $q_k = \alpha
%    v_{k-1} + \beta v_k$ for some scalars $\alpha$ and $\beta$, we
%    would assume that $\norm{q_k} = \norm{[\alpha \ \; \beta]}$.
%    Even though their relative accuracies may diminish, the resulting
%    norm estimates help determine a safe point of
%    termination.  \end{assumption}

If the skew symmetric process were  forced on a skew Hermitian
matrix, the resultant $V_k$ would \emph{not} be orthogonal. Instead,
we multiply $Ax \approx b$ by $i$ on both sides to yield a Hermitian
problem since $(iA)^* = \conj{i} A^* = i A$. This simple
transformation by a scalar multiplication\footnote{Multiplying  by $-i$
  works equally well\red{,} but without loss of generality, we use $i$.}
%NOTE: SH matrix always have eigenvalues in conjugate pairs. Thus, iA
%always have positive and negative eigenvalues of the same magnitude.
preserves the \red{conditioning since} $\kappa(A)=\kappa(iA)$ and
allows us to adapt the original Hermitian Lanczos process with $v_0
\equiv 0$, $\beta_1 v_1= i b$, followed by
\begin{equation} \label{eq:sh-avk}
    p_k =  iA v_k, \qquad \alpha_k = v_k^* p_k, \qquad
    \beta_{k+1}v_{k+1} = p_k-\alpha_kv_k-\beta_kv_{k-1}.
\end{equation}
Its matrix form is the same as~\eqref{eq:avk} except that the first
equation is $i A V_k = V_{k+1}\underTk$.

%------------------------------------------------ 
\subsection{Properties of the Lanczos processes} 
\label{sec:Lanproperties}
%------------------------------------------------ 

The following properties of the Lanczos processes are notable:
\begin{enumerate}
\item If~$A$ and~$b$ \red{are} real, then the Saunders
  process~\eqref{eq:savk} \red{for a complex symmetric system} reduces
  to the symmetric Lanczos process.
\item The complex and skew symmetric properties of~$A$ carry over
  to~$T_k$ by the Lanczos processes~\eqref{eq:savk}
  and~\eqref{eq:savk-ss}\red{,} respectively.  From the skew Hermitian
  process~\eqref{eq:sh-avk}, $T_k$ is symmetric.
\item The skew symmetric Lanczos process~\eqref{eq:savk-ss} is only
  two-term recurrent.
\item In~\eqref{eq:sh-avk}, there are two ways to form $p_k$: $p_k =
  (iA) v_k$ or $p_k = A(iv_k)$.  One may be cheaper than the other. If
  $A$ is dense, $iA$ takes $\mathcal{O}(n^2)$ scalar multiplications
  and storage. If $A$ is sparse or structured as in the case of
  Toeplitz, $iA$ just takes $\mathcal{O}(n)$ multiplications. In
  contrast, $iv_k$ takes $n\red{\wp}$ multiplications, where~$\red{\wp}$ is
  theoretically bounded by the number of distinct nonzero eigenvalues
  of~$A$\red{;} but in practice~$\red{\wp}$ could be an integer multiple
  of~$n$.
\item While the skew Hermitian Lanczos process~\eqref{eq:sh-avk} is
  applicable to a skew symmetric problem, it involves complex
  arithmetic and is thus computationally more costly than the skew
  symmetric Lanczos process with a real \red{vector} $b$.
\item If $A$ is changed to $A-\sigma I$ for some scalar shift
  $\sigma$, \red{then} $T_k$ becomes $T_k - \sigma I$\red{,} and $V_k$ is unaltered,
  showing that singular systems are commonplace.  Shifted problems
  appear in inverse iteration or Rayleigh quotient iteration. The \red{Saunders and}
  Lanczos \red{frameworks efficiently handle} shifted problems.
\item Shifted skew symmetric matrices are not skew
  symmetric. This notion also applies to the case of shifted skew
  Hermitian matrices. Nevertheless they arise often in Toeplitz
  problems~\cite{CJ91,CJ07}.
\item For the skew Lanczos processes, the $k$th Krylov subspace
  generated by $A$ and $b$ is defined to be $\mathcal{K}_k(A,b) =
  \range(V_k) = \Span \{b,Ab, \dots, A^{k-1}b\}$.  For the Saunders
  process, we have a \emph{modified} Krylov subspace~\cite{SSY88} that
  we call the \emph{Saunders subspace}, \red{$%\range(\conj{V}\!_k) =
 \mathcal{S}_k(A,b) \equiv  \mathcal{K}_{k_1}(A\conj{A},b) \oplus
  \mathcal{K}_{k_2}(A\conj{A},A\conj{b})$},
 % =\Span\{b,A\conj{A}b,(A\conj{A})^2b,\ldots\} \cup
 % \Span\{A\conj{b},A\conj{A}A\conj{b},(A\conj{A})^2A\conj{b},\ldots\}$
where \red{$\oplus$ is the direct-sum operator,} $k_1+k_2 = k$\red{,} and $0 \le k_1-k_2\le1 $.
%\red{Given two complementary
%subspaces $X$ and $Y$ of $\mathbb{C}^n$ such that $X\cap Y=\{0\}$, 
%we define the direct sum of $X$ and $Y$
%as $X \oplus Y = \{x+y: x\in X, y\in Y\}.$}
 
\item \label{prop:fullrank} $\underTk$ has full column rank $k$ for
  all $k < \ell$ \red{because} $\beta_1,\dots,\beta_{k+1} > 0$.
\end{enumerate}

\smallskip
\begin{theorem} $T_\ell$ is nonsingular if and only if $b \in \range(A)$.  
Furthermore,
$\rank(T_\ell) = \ell-1$ in the case $b \notin \range(A)$.
      %\begin{cases}
      %     {\ell}    & \text{if } b \in \range(A)
       %\\  {\ell-1}    & \text{otherwise}
       %\end{cases}.
\begin{proof}
  We prove below for $A$ complex symmetric. The proofs are similar
  for the skew symmetric and skew Hermitian cases.
  
  We use $A \overline{V}_\ell = V_\ell T_\ell$ twice.  First, if
  $T_\ell$ is nonsingular, we can solve $T_\ell y_\ell = \beta_1 e_1$
  and then $A \overline{V}_\ell y_\ell = V_\ell T_\ell y_\ell = V_\ell
  \beta_1 e_1 = b$.  Conversely, if $b \in \range(A)$, then
  $\range(\overline{V}_\ell) \subseteq \range(\overline{A}) =
  \range(A^*)$.  Suppose $T_\ell$ is singular. Then there exists $z
  \ne 0$ such that $T_\ell z=0$ and thus $V_\ell T_\ell z = A
  \overline{V}_\ell z=0$. That is, $0 \ne \overline{V}_\ell z \in
  \Null(A)$.  But this is impossible because $\overline{V}_\ell z \in
  \range(A^*)$ and $\Null(A) \cap \range(A^*) =\{ 0 \}$.  Thus
  $T_\ell$ must be nonsingular.

  If $b \notin \range(A)$, $T_\ell = \bmat{\underline{T_{\ell-1}}&
      \begin{smallmatrix}\beta_\ell e_{\ell-1}
                         \\ \alpha_\ell
      \end{smallmatrix}}$
is singular. It follows that $\ell > \rank(T_\ell) \ge
\rank(\underline{T_{\ell-1}}) = \ell-1$ since $\rank(\underline{T_k})
= k$ for all $k<\ell$. Therefore  $\rank(T_\ell) = \ell-1$. % \endproof
\end{proof}
\label{thm:rankTk}
\end{theorem}

%------------------------------------------------ 
\subsection{QLP decompositions for singular matrices}  
  \label{sec:QLPreview}
%------------------------------------------------ 

Here we generalize, from real to complex, the matrix decomposition
\emph{pivoted QLP} by Stewart in 1999~\cite{S99}.\footnote{QLP is a
  special case of the \ULV decomposition, also by
  Stewart~\cite{S93,HC03}.}  It is equivalent to two consecutive \QR
factorizations with column interchanges, first on $A$, then on
\red{$\smat{R & S}^*$}:
\begin{equation}  \label{qlpeqn1}
   Q_R A \Pi_R = \bmat{ R & S \\ 0&0 }, \qquad
   Q_L \bmat{R^* & 0 \\ S^* & 0} \Pi_L = \bmat{\hat{R} & 0 \\ 0&0},
\end{equation}
giving \emph{nonnegative} diagonal elements, where $\Pi_R$ and $\Pi_L$
are (real) permutations chosen to maximize the next diagonal element
of $R$ and $\hat{R}$ at each stage.  This gives
\begin{equation*}
A = QLP, \qquad
  Q = Q_R^* \Pi_L, \qquad
  L = \bmat{\hat{R}^* & 0 \\ 0&0}, \qquad
  P =  Q_L \Pi_R^T,
\label{qlpeqn2}
\end{equation*}
with $Q$ and $P$ orthonormal.  Stewart demonstrated that the diagonal
elements of $L$ (the \emph{$L$-values}) give better singular-value
estimates than those of $R$ (the \emph{$R$-values}), and the accuracy
is particularly good for the extreme singular values $\sigma_1$ and
$\sigma_n$:
\begin{equation}  \label{qlpeqn1a}
  R_{ii} \simeq \sigma_i, \quad L_{ii} \simeq \sigma_i, \quad
  \sigma_1 \ge \max_i L_{ii} \ge \max_i R_{ii}, \quad \min_i R_{ii}
  \ge \min_i L_{ii} \ge \sigma_n.\!\!
\end{equation}
The first permutation $\Pi_R$ in pivoted \QLP is important.  The main
purpose of the second permutation $\Pi_L$ is to ensure that the
$L$-values present themselves in decreasing order, which is not always
necessary.  If $\Pi_R = \Pi_L = I$, it is simply called the \emph{\QLP
  decomposition}, which is applied to each $T_k$ from the Lanczos
processes (Section~\ref{sec:Lanczos}) in \MINRESQLP.

%------------------------------------------------ 
\subsection{Householder \red{reflectors}}  
\label{sec:reflector}
%------------------------------------------------ 

  Givens rotations are often used to selectively annihilate  matrix
  elements.  Householder reflectors~\cite{TB} of the following form
  may be considered the \emph{Hermitian} counterpart of Givens
  rotations:
\begin{equation*}
 Q_{i,j} \!\equiv\! \mbox{\footnotesize
    $\bmat{1       & \cdots & 0       &  \cdots & 0          &  \cdots & 0 \;
        \\ \vdots  & \ddots & \vdots  &         & \vdots     &         & \vdots
        \\   0     & \cdots &  c      & \cdots  & \phantom-s & \cdots  & 0
        \\ \vdots  &        & \vdots  &  \ddots & \vdots     &         & \vdots
        \\    0    & \cdots & \conj{s}& \cdots  & -c         & \cdots  & 0
        \\ \vdots  &        & \vdots  &         & \vdots     &  \ddots & \vdots
        \\    0    & \cdots & 0       & \cdots  & 0          & \cdots  & 1
        }
    $},
\end{equation*}
where the subscripts indicate the positions of $c=\cos(\theta) \in
\mathbb{R}$ and $s=\sin(\theta) \in \mathbb{C}$ for some angle
$\theta$. They are orthogonal\red{,} and $Q_{i,j}^2 = I$ as \red{for} any
reflector, meaning $Q_{i,j}$ is its own inverse. Thus $c^2
+|s|^2=1$. We often use the shorthand $Q_{i,j} = \smat{c & \phantom-s
  \\[.8ex] \conj{s} & -c}$.
%and apply $Q_{i,j} \smat{a \\ b }$ t \smat{r \\ 0}$.

\begin{comment}
\begin{remark} \label{rem:ns}
  In~\eqref{qlpeqn3a}, $\beta_{k+1}>0$ implies $R_k$ and $L_k$ are
  nonsingular (and $k<\ell$).
  %Suppose $\beta_{k+1}=0$.  In
  %section~\ref{sec:Lanczos}, $b \in \mathcal{R}(A)$ if and only if
  %$v_1,\ldots, v_k \perp \mathcal{N}(A)$,
  %so $b \in \mathcal{R}(A)$ if and only if $T_k$ is nonsingular (since
  %$AV_k=V_kT_k$).  If $T_\ell$ is nonsingular then $T_\ell y =
  \beta_1 e_1$ has a unique solution and $AV_ky_k = V_kT_ky_k = b$,
  where $x_k = V_ky_k\perp \Null(A)$, and $r_k = 0$; otherwise $b
  \not\in \range(A)$, $T_k$ is singular, and the residual is nonzero.
\end{remark}
\end{comment}

\bigskip%\bigskip 

In the next few sections we extend \MINRES and
\MINRESQLP to solving complex symmetric problems~\eqref{eqn4b}. Thus
we tag the algorithms with ``CS-''.  The discussion and results can be
easily adapted to the skew symmetric and skew Hermitian cases\red{,} and so
we do not go into details.  In fact, the skew Hermitian problems can
be solved by the implementations~\cite{MinresqlpMatlab, MinresqlpF90}
of \MINRES and \MINRESQLP for Hermitian problems.  For example, we can
call the \red{\Matlab} solvers by \texttt{x = minres(i * A, i * b)} and
\texttt{x = minresqlp(i * A, i * b)} \red{to achieve} code reuse immediately.

%%%%%%%%%%%%%%%%%%%%%%%%%%%%%%%%%%%%%%%%%%%%%%%%%%%%%%%%%%%%%%%%
\section{\CSMINRES standalone}  \label{sec:csminres-standalone}
%%%%%%%%%%%%%%%%%%%%%%%%%%%%%%%%%%%%%%%%%%%%%%%%%%%%%%%%%%%%%%%%
\CSMINRES is a natural way of using the complex symmetric Lanczos
process~\eqref{eq:savk}
%\eqref{eq:avk}
to solve~\eqref{eqn4b}.  For $k < \ell$, if $x_k =
\conj{V}\!_ky_k$ for some vector $y_k$, the associated residual is
\begin{equation}  \label{eqn:rk}
   r_k \equiv b-Ax_k = b - A \conj{V}\!_k y_k = \beta_1 v_1 - V_{k+1}
   \underTk y_k = V_{k+1} (\beta_1 e_1 - \underTk y_k).
\end{equation}
\red{In order to} make $r_k$ small,  $\beta_1 e_1 - \underTk y_k$
should be small.  At this iteration $k$, \CSMINRES minimizes the
residual subject to $x_k\in\range(\overline{V}_k)$ by choosing
\begin{equation}  \label{eqn:LSsubprob}
  y_k = \arg\min_{y \in \mathbb{C}^k}
        \norm{\underTk y - \beta_1 e_1}.
\end{equation}
By Theorem~\ref{thm:rankTk}, $\underTk$ has full column rank\red{,} and the
above is a nonsingular problem.

%------------------------------------------------ 
\subsection{QR factorization of $\underTk$}
\label{sec:qrfac}
%------------------------------------------------ 

We apply an expanding \QR factorization to the
subproblem~\eqref{eqn:LSsubprob} by $Q_0 \equiv 1$ and
\begin{equation}  \label{QRfac}
  Q_{k,k+1}\!\equiv\!
   \bmat{ % I_{k-1}&&\\
              c_k &   \!\!\phantom-s_k
        \\    \conj{s}_k &   \!-c_k
        },
                 \quad
  Q_k \!\equiv\! Q_{k,k+1}
                 \bmat{Q_{k-1} \\ & \!1}\!,
                 \quad
                 Q_k \bmat{\underTk & \beta_1 e_1}
      \!=\!      \bmat{ R_k & t_k \\ 0 & \phi_{k}}\!,
\end{equation}
where $c_k$ and $s_k$ form the Householder reflector $Q_{k,k+1}$ that
annihilates $\beta_{k+1}$ in $\underTk$ to give upper-tridiagonal
$R_k$, with $R_k$ and $t_k$ being unaltered in later iterations.  We
can \red{state} the last expression in~\eqref{QRfac} in terms of its
elements for further analysis:
%for $1 \le k < \ell$, 
\begin{equation}  \label{QRfac2}
  \bmat{\,R_k\, \\ 0} \equiv \mbox{\small
$\bmat{
      \gamma_1 & \delta_2 & \epsilon_3
\\                   & \gamma_2^{(2)} & \delta_3^{(2)} & \ddots
\\                   &                & \ddots         & \ddots & \epsilon_k
\\                   &                &                & \ddots & \delta_k^{(2)}
\\                   &                &                &        & \gamma_k^{(2)}
\\                   &                &                &        & 0}$},
  \quad
  \bmat{t_k \\ \phi_{k}} \equiv
  \bmat{\tau_1 \\ \tau_2 \\ \vdots\\ \vdots\\ \tau_k \\ \phi_{k}}
  = \beta_1
  \bmat{c_1 \\ \conj{s}_1 c_2 \\ \vdots \\ \vdots
     \\ \conj{s}_1 \cdots \conj{s}_{k-1}c_k
     \\ \conj{s}_1 \cdots \conj{s}_{k-1}\conj{s}_k}
\end{equation}
(where the superscripts are defined in Section~\ref{sec:notation}).
With $\phi_1 \equiv \beta_1>0$, the full action of $Q_{k,k+1}$ in
\eqref{QRfac}, including its effect on later columns of $T_i$, $k < i
\le \ell$, is described~by
\begin{equation}  \label{min7}
  \bmat{          c_k & \!\!\!\phantom-s_k
        \\ \conj{s}_k & \!\!-c_k}
  \bmat{\begin{matrix}
           \gamma_k & \delta_{k+1} & 0
        \\ \beta_{k+1}    & \alpha_{k+1}       & \beta_{k+2}
        \end{matrix}
        & \biggm| &
        \begin{matrix} \phi_{k-1} \\ 0 \end{matrix}
       }
=
  \bmat{\begin{matrix}
           \gamma_k^{(2)} & \delta_{k+1}^{(2)} & \epsilon_{k+2}
        \\ 0              & \gamma_{k+1} & \delta_{k+2}
        \end{matrix}
        & \biggm| &
        \begin{matrix} \tau_k \\ \phi_{k} \end{matrix}
       }.
\end{equation}
Thus for each $j \le k < \ell$ we have $s_j \gamma_j^{(2)} =
\beta_{j+1} > 0$, giving $\gamma_1$, $\gamma_j^{(2)} \ne 0$, and
therefore each $R_j$ is nonsingular. Also, $\tau_k = \phi_{k-1} c_k$
and $\phi_{k} = \phi_{k-1} \conj{s}_k \ne 0$.  Hence
from~\eqref{eqn:rk}--\eqref{QRfac}, we \red{obtain} the following short
recurrence relation for the residual norm:
\begin{align}
   \label{eq:normrk}
   \norm{r_k} = \norm{\underTk y_k - \beta_1 e_1} = |\phi_k|
   \quad\Rightarrow\quad \norm{r_k} = \norm{r_{k-1}} |\conj{s}_k| 
   = \norm{r_{k-1}} |s_k|,
\end{align}
which is monotonically decreasing and tending to zero if $Ax = b$ is
compatible.

%------------------------------------------------
\subsection{Solving the subproblem}
%------------------------------------------------

% When $\beta_{k+1}>0$, $R_k$ is nonsingular
When $k < \ell$, % $\underTk$ and $R_k$ have full rank and the uniquea
a solution of~\eqref{eqn:LSsubprob} satisfies $R_k y_k = t_k$.
Instead of solving for $y_k$, \CSMINRES solves $R_k\T D_k\T = V_k^*$
by forward substitution, obtaining the last column $d_k$ of $D_k$ at
iteration $k$. This basis generation process can be summarized as
\begin{align}
  % R_k\T D_k\T = \conj{V}_k\T\ :\ &
   \left\{
     \begin{array}{l}
       d_1 = \conj{v}_1/\gamma_1,\quad d_2 =
       (\conj{v}_2-\delta_2d_1)/\gamma_2^{(2)}, \\[1ex] d_k =
       ({\conj{v}_k - \delta_k^{(2)} d_{k-1} - \epsilon_k d_{k-2}}) /
       {\gamma_k^{(2)}}.
       %\quad j=3,\ldots,k.
     \end{array}
   \right.
   \label{eq:rdeqv}
\end{align}
At the same time, \CSMINRES updates $x_k$ via $x_0 \equiv 0$ and
\begin{equation}  \label{minresxk}
    x_k = \conj{V}\!_k y_k = D_k R_k y_k = D_k t_k
  = x_{k-1} + \tau_k  d_k.%,\quad \tau_k \equiv e_k^Tt_k.
\end{equation}

%------------------------------------------------
\subsection{Termination}
%------------------------------------------------

When $k=\ell$, we can form $T_\ell$\red{,} but nothing else expands.  In
place of~\eqref{eqn:rk} and~\eqref{QRfac} we have \( r_\ell = V_\ell
(\beta_1 e_1 - T_\ell y_\ell) \) and \( Q_{\ell-1} \bmat{T_\ell &
  \beta_1 e_1} = \bmat{R_\ell & t_\ell} \). It is natural to solve for
$y_\ell$ in the subproblem
\begin{align}  \label{eqn:LSsubprob-ell}
    \min_{y_\ell \in \mathbb{C}^{\ell}} 
    \norm{T_\ell y_\ell - \beta_1 e_1}
    \quad \equiv \quad
    \min_{y_\ell \in \mathbb{C}^{\ell}} 
    \norm{R_\ell y_\ell - t_\ell}.
\end{align}
\red{Two} cases \red{must be considered}:
\begin{enumerate}
\item If $T_\ell$ is nonsingular, $R_\ell y_\ell = t_\ell$ has a
  unique solution.  Since $A\conj{V}\!_\ell y_\ell = V_\ell T_\ell
  y_\ell = b$, the problem $Ax=b$ is compatible and solved by $x_\ell
  = \conj{V}\!_\ell y_\ell $ with residual $r_\ell=0$.
  %(the system is compatible, even if $A$ is singular).  
  Theorem~\ref{theorem-singular-compatible}   proves that
  $x_\ell = x^\dagger$, assuring us that \CSMINRES is a useful solver 
for compatible linear systems even if $A$ is singular. 

\item If $T_\ell$ is singular, $A$ and $R_\ell$ are singular
  ($R_{\ell\ell}=0$)\red{,} and both $Ax = b$ and $R_\ell y_\ell = t_\ell$
  are incompatible. The optimal residual vector is unique, but
  infinitely many solutions give that residual. \CSMINRES sets the
  last element of $y_\ell$ to be zero.  The final point and residual
  stay as $x_{\ell-1}$ and $r_{\ell-1}$ with $\norm{r_{\ell-1}} =
  |\phi_{\ell-1}| = \beta_1 |s_1 |\cdots| s_{\ell-1}| > 0$.
  Theorem~\ref{theorem-singular-incompatible}   proves that
  $x_{\ell-1}$ is a \LS solution of $Ax \approx b$ (but not
  necessarily $x^\dagger$).
  %Theorem \ref{theorem-MINRES-QLP} proves that the \CSMINRESQLP point
  %$x_\ell = \conj{V}_\ell y_\ell^\dagger = x^\dagger$, where
  %$y_\ell^\dagger$ is the minimum-length solution
  %of~\eqref{eqn:LSsubprob-ell}.
\end{enumerate}

\begin{comment}
\begin{remark}
\label{rem:stop1}
If $\gamma_k = 0$ then $T_k$ is singular.  If $k < \ell$ %$\beta_{k+1}
> 0$ we have $s_k=1$ and $\norm{r_k} = \norm{r_{k-1}}$ (not a strict
decrease), but this cannot happen twice in a row.
%If $\gamma_k = 0$ and $\beta_{k+1}=0$, then $Q_{k-1,k}$ terminates
%the decomposition, the equivalent of~\eqref{QRfac2} has
%$\smash{\gamma_k^{(2)}} = 0$, and the final residual stays as
%$r_{k-1}$ with $\norm{r_{k-1}} = \phi_{k-1} = \beta_1 s_1 \cdots
%s_{k-1} \ne 0$, see Remark~\ref{rem:ns}, but $y_k$ is not unique in
%\eqref{eqn:LSsubprob}.
If $k=\ell-1$ and $\gamma_{k+1} = \gamma_\ell = 0$ in~\eqref{min7},
then the final point and residual stay as $x_{\ell-1}$ and
$r_{\ell-1}$ with $\norm{r_{\ell-1}} = \phi_{\ell-1} = \beta_1 s_1
\cdots s_{\ell-1} > 0$.
%, see Remark~\ref{rem:ns},
%but $y_\ell$ is not unique in~\eqref{eqn:LSsubprob}.
\end{remark}
\begin{remark} \label{rem:stop1}
  If $k < \ell $ and $T_k$ is singular, we have $\gamma_k = 0$, $s_k =
  1$, and $\norm{r_k} = \norm{r_{k-1}}$ (not a strict decrease), but
  this cannot happen twice in a row
  (cf.~Section~\ref{sec:Lanproperties}).
\end{remark}

\begin{remark} \label{rem:stop2}
  If $T_\ell$ is singular, 
\end{remark}
\end{comment}

%-------------------------------------------------------------------
%\subsection{Properties of \CSMINRES}   \label{sec:CSMINRESdetails}
%-------------------------------------------------------------------

%\subsection{Compatible systems}

\smallskip
\begin{theorem}  \label{theorem-singular-compatible}
  If $b \in \range(A)$, the final \CSMINRES point $x_\ell=x^\dagger$
  and $r_\ell=0$.
\end{theorem}
\begin{proof}
%This follows from Remark~\ref{rem:ns}.
If $b \in \range(A)$, the Lanczos process gives $A\conj{V}\!_\ell =
V_\ell T_\ell$ with nonsingular $T_\ell$, and \CSMINRES terminates
with $Ax_\ell=b$ and $x_\ell = \conj{V}\!_\ell y_\ell = A^*
q=\conj{A}q$, where $q = V_\ell \conj{T}\!_\ell\inv y_\ell$.  If some
other point $\xhat$ satisfies $A\xhat=b$, let $p = \xhat - x_\ell$.
We have $Ap=0$ and $x_\ell^* p = q^* Ap = 0$.  Hence $\norm{\xhat}^2 =
\norm{x_\ell + p}^2 = \norm{x_\ell}^2 + 2 x_\ell^* p + \norm{p}^2 \ge
\norm{x_\ell}^2$. \red{Thus  $x_\ell=x^\dagger$.} 
\red{Since $\beta_{\ell+1}=0$, $s_k=0$ in~\eqref{min7}. By~\eqref{eq:normrk}, $\norm{r_\ell}=0$ and $r_\ell = b-Ax_\ell = 0$}.
\end{proof}

%------------------------------------------------
%\subsection{Incompatible systems}
%------------------------------------------------

%For a singular \LS problem $Ax \approx b$, the optimal residual vector
%$\rhat$ is unique, but infinitely many solutions $x$ give that
%residual.  \CSMINRES finds a least-squares solution (with optimal
%residual) but not necessarily the minimum-length solution.

\smallskip
\begin{theorem}   \label{theorem-singular-incompatible}
  If $b \notin \range(A)$, then $\norm{Ar_{\ell-1}} = 0$\red{,} and the
  \CSMINRES $x_{\ell-1}$ is an \LS solution.
% (but not necessarily $x^\dagger$).
\end{theorem}

\begin{proof}
Since $b \notin \range(A)$, $T_\ell$ is singular and $R_{\ell\ell} =
\gamma_\ell = 0$.  By Lemma~\ref{minreslemma2},
$A^*(Ax_{\ell-1}-b) = -\conj{A}r_{\ell-1} = -\norm{r_{\ell-1}} \gamma_\ell
v_\ell = 0$.  Thus $x_{\ell-1}$ is an \LS solution.
\end{proof}

\begin{comment}
\smallskip
We illustrate this result with a small example:
\smallskip
\begin{example} \label{minresCounterEg}
  Let $A = i \; \diag(1,1,0)$ and $b=i \; e$.  The minimum-length
  solution to $Ax \approx b$ is $x^\dagger = [ 1 \ 1 \ 0 ]\T$ with
  residual $\rhat = b - Ax^\dagger = e_3$ and $A\rhat = 0$.  \CSMINRES
  returns the solution $x^\sharp = e$ (with residual $r^\sharp = b -
  Ax^\sharp= e_3 = \rhat$ and $Ar^\sharp = 0$).
\end{example}
\end{comment}

%%%%%%%%%%%%%%%%%%%%%%%%%%%%%%%%%%%%%%%%%%%%%%%%%%%%%%%%%%%%%%%%
\section{\CSMINRESQLP standalone}  \label{sec:CSMINRESQLP}
%%%%%%%%%%%%%%%%%%%%%%%%%%%%%%%%%%%%%%%%%%%%%%%%%%%%%%%%%%%%%%%%

In this section we develop \CSMINRESQLP for solving ill-conditioned or
singular symmetric systems.  The Lanczos framework is the same as in
\CSMINRES, and QR factorization is applied to $\underTk$ in
subproblem~\eqref{eqn:LSsubprob} for all $k < \ell$; see
Section~\ref{sec:qrfac}. By Theorem~\ref{thm:rankTk} and
Property~\ref{prop:fullrank} in Section~\ref{sec:Lanproperties},
$\rank(\underTk) = k$ for all $k<\ell$ and $\rank(T_\ell) \ge \ell-1$.
\CSMINRESQLP handles $T_\ell$ in~\eqref{eqn:LSsubprob-ell} with extra
care to \emph{constructively} reveal $\rank(T_\ell)$
%      or $T_k$ is ill-conditioned.
%%% Can't say this when lots of T_k could be singular.
%%% We could say
%      or earlier $\underTk$ are ill-conditioned.
%%% but it confuses things.
 via a \QLP decomposition, so it can compute the minimum-length
 solution of the following subproblem~instead
 of~\eqref{eqn:LSsubprob-ell}:
\begin{align}  \label{eqn:LSsubprob-ell-2}
   & \min \norm{y_\ell}_2 \quad \text{s.t.} \quad y_\ell \in 
  \arg\min_{y_\ell \in \mathbb{C}^{\ell}} 
  \norm{T_\ell y_\ell - \beta_1 e_1}.
% \\  \quad \equiv \quad
%   & \min \norm{y_\ell}_2 \quad \text{s.t.} \quad y_\ell \in 
%   \arg\min_{y_\ell \in \mathbb{C}^{n}} \norm{L_\ell y_\ell - t_\ell}.
%   \nonumber
\end{align}
Thus \CSMINRESQLP also applies the \QLP decomposition on $\underTk$
in~\eqref{eqn:LSsubprob} for all $k < \ell$.
\begin{comment}
%------------------------------------------------
\subsection{\CSMINRES subproblem}
%------------------------------------------------

If $k < \ell$, $\underTk$ has full column rank and $y_k$ solves the
same subproblem $\min \norm{\underTk y - \beta_1 e_1}$
\eqref{eqn:LSsubprob} as \CSMINRES.  If $A$ is singular and $b \not\in
\range(A)$, $T_\ell$ is singular and we choose $y_\ell$ to solve the
subproblem
\begin{equation}  \label{qlpsubproblemk}
  \min \norm{y} \quad\text{s.t.}\quad y \in \arg\min_{y \in
    \mathbb{R}^\ell} \norm{T_\ell y- \beta_1 e_1}.
\end{equation}
For all $k \le \ell$ the solutions $y_k$ define $x_k = V_k y_k$ % \in
\mathcal{K}_k (A,b)$ as the $k$th approximation to $x$.  As usual,
$y_k$ is not actually computed because all its elements change when
$k$ increases.
\end{comment}

%------------------------------------------------
\subsection{QLP factorization of $\underTk$}
%------------------------------------------------

In \CSMINRESQLP, the \QR factorization~\eqref{QRfac} of $\underTk$ is
followed by an \LQ factorization of $R_k$:
\begin{equation}  \label{qlpeqn3a}
  Q_k \underTk =  \bmat{R_k \\ 0}, \qquad
  R_k P_k = L_k,  \qquad \textrm{so that}\quad
  Q_k \underTk P_k =  \bmat{L_k \\ 0},
\end{equation}
where $Q_k$ and $P_k$ are orthogonal, $R_k$ is upper tridiagonal, and
$L_k$ is lower tridiagonal.  When $k < \ell$, both $R_k$ and $L_k$ are
nonsingular.  The \QLP decomposition of each $\underTk$ \red{is} performed
without permutations, and the left and right reflectors \red{are}
interleaved~\cite{S99} in order to ensure inexpensive updating of the
factors as $k$ increases.  The desired rank-revealing
properties~\eqref{qlpeqn1a} are retained in the last iteration when
$k=\ell$.

%For $k < \ell$, the QLP decomposition of $\underTk$~\eqref{qlpeqn3a}
%\begin{equation}  \label{qlpeqn3}
%  Q_k \underTk = \bmat{R_k \\ 0},
%  \qquad
%  R_k P_k = L_k,
%  \qquad
%  Q_k \underTk P_k = \bmat{L_k \\ 0},
%\end{equation}
%gives nonsingular tridiagonal $R_k$ and $L_k$.  
We elaborate on interleaved QLP here. As in \CSMINRES, $Q_k$
in~\eqref{qlpeqn3a} is a product of Householder reflectors\red{;}
see~\eqref{QRfac} and~\eqref{min7}\red{.} $P_k$ involves a product of
\emph{pairs} of Householder reflectors:
\begin{equation*}
  Q_k = Q_{k,k+1}\ \cdots \ Q_{3,4} \ \ Q_{2,3} \ \ Q_{1,2},\qquad
  P_k  =  P_{1,2} \ \ P_{1,3} P_{2,3}
  \ \cdots \ \ P_{k-2,k} P_{k-1,k}.
\end{equation*}
For \CSMINRESQLP to be efficient, in the $k$th iteration ($k\ge 3$)
the application of the left reflector $Q_{k,k+1}$ is followed
immediately by the right reflectors $P_{k-2,k}, P_{k-1,k}$, so that
only the last $3 \times 3$ bottom right submatrix of  $\underTk$ is changed.  
These ideas can be understood more easily from
%Figure~\ref{QLPfig} and 
the following compact form, which represents
the actions of right reflectors on $R_k$ obtained from~\eqref{min7}: 
%(additional to $Q_{k,k+1}$~\eqref{min7}):
\begin{align}
\label{qlpRightRef}
 &\hspace*{13pt}
  \bmat{\gamma_{k-2}^{(5)} &                    & \epsilon_k
     \\ \vartheta_{k-1}    & \gamma_{k-1}^{(4)} & \delta_k^{(2)}
     \\                    &                    & \gamma_{k}^{(2)}
       }
  \bmat{c_{k2} &   & \!\!\!\phantom-s_{k2}
     \\        & 1 &
     \\ \conj{s}_{k2} &   & \!\!-c_{k2}
       }
  \bmat{ 1
     \\    & c_{k3} & \!\!\!\phantom-s_{k3}
     \\    & \conj{s}_{k3} & \!\!-c_{k3}
       }
       \nonumber
\\ &=
  \bmat{\gamma_{k-2}^{(6)}
     \\ \vartheta_{k-1}^{(2)} & \gamma_{k-1}^{(4)}  & \delta_k^{(3)}
     \\ \eta_{k}              &                     & \gamma_{k}^{(3)}
       }
  \bmat{ 1
     \\    & c_{k3} & \!\!\!\phantom-s_{k3}
     \\    & \conj{s}_{k3} & \!\!-c_{k3}
       }
= \bmat{\gamma_{k-2}^{(6)}
     \\ \vartheta_{k-1}^{(2)} & \gamma_{k-1}^{(5)}  &
     \\ \eta_{k}              & \vartheta_{k}       & \gamma_{k}^{(4)}
       }.
\end{align}

\subsection{Solving the subproblem}
%------------------------------------------------

With $y_k=P_k u_k$, subproblem~\eqref{eqn:LSsubprob} after \QLP
factorization of $\underTk$ becomes
%$L_k u_k = t_k$ and $y_k = P_k u_k$.
%equivalent to
\begin{equation}  \label{Lsubproblem}
  u_k = \arg \min_{u \in \mathbb{C}^k}
            \,\normm{\bmat{L_k \\ 0} u - \bmat{t_k \\ \phi_{k}}},
\end{equation}
where $t_k$ and $\phi_{k}$ are as in~\eqref{QRfac}.  At the
\textit{start} of iteration $k$, the first $k\!-\!3$ elements of
$u_k$, denoted by $\mu_j$ for $j \le k\!-\!3$, are known from previous
iterations.
%see the 10th matrix in Figure~\ref{QLPfig}.
%\begin{enumerate}
%\item
We need to solve \red{for} only the last three components of $u_k$ from the
bottom three equations of $L_k u_k = t_k$:
\begin{equation}
  \bmat{\gamma_{k-2}^{(6)}
     \\ \vartheta_{k-1}^{(2)} & \gamma_{k-1}^{(5)}  &
     \\ \eta_k               & \vartheta_k       & \gamma_k^{(4)}
       }
  \bmat{\mu_{k-2}^{(3)}
     \\ \mu_{k-1}^{(2)}
     \\ \mu_k
       }
= \bmat{\tau_{k-2}^{(2)}
     \\ \tau_{k-1}^{(2)}
     \\ \tau_k
       }
\equiv
  \bmat{\tau_{k-2} - \eta_{k-2} \mu_{k-4}^{(4)}
                   - \vartheta_{k-2} \mu_{k-3}^{(3)}
     \\ \tau_{k-1} - \eta_{k-1} \mu_{k-3}^{(3)}
     \\ \tau_k
       }.
\end{equation}
When $k < \ell$, $\underTk$ has full column rank, and so do $L_k$ and
the above $3\times 3$ triangular matrix.
\CSMINRESQLP obtains the same solution as \CSMINRES,
but by a different process (and with different rounding errors).
The \CSMINRESQLP estimate of $x$ is   \( x_k =
\conj{V}\!_ky_k = \conj{V}\!_kP_ku_k = W_ku_k, \) with theoretically
orthonormal $W_k\equiv \conj{V}\!_kP_k$, where
\begin{align}
  W_k %\equiv \conj{V}\!_k P_k
         &=  \bmat{ \conj{V}\!_{k-1} P_{k-1} &  \conj{v}_k } P_{k-2,k} P_{k-1,k}    
             \label{eq:wvp}
 \\      &=  \bmat{ W_{k-3}^{(4)} &  w_{k-2}^{(3)} & w_{k-1}^{(2)} & \conj{v}_k} 
         P_{k-2,k} P_{k-1,k} \nonumber
 \\      &\equiv \bmat{ W_{k-3}^{(4)} & w_{k-2}^{(4)} & w_{k-1}^{(3)} & w_k^{(2)}}.
         \nonumber
 %\\  W_k^* W_k &= I_k, \qquad \range(W_k)=\mathcal{K}_k(A,\conj{b}),          
 % \label{W_ortho}
\end{align}
\red{Finally}, we update å$x_{k-2}$ and compute $x_k$ by short-recurrence
orthogonal steps (using only the last three columns  
of $W_k$):
\begin{align}
x_{k-2}^{(2)}  &= x_{k-3}^{(2)} + w_{k-2}^{(4)} \mu_{k-2}^{(3)}
\text{, where } x_{k-3}^{(2)} \equiv W_{k-3}^{(4)} u_{k-3}^{(3)}, 
 \label{qlpeqnsol1}
\\     x_k    &= x_{k-2}^{(2)} + w_{k-1}^{(3)} \mu_{k-1}^{(2)}
                      + w_k    ^{(2)} \mu_k.      \label{qlpeqnsol2}
\end{align}

%------------------------------------------------
\subsection{Termination}  \label{sec:term}
%------------------------------------------------

When $k=\ell$ and $y_\ell = P_\ell u_\ell$, the final
subproblem~\eqref{eqn:LSsubprob-ell-2} becomes
\begin{align}  \label{eqn:LSsubprob-ell-3}
   & \min \norm{u_\ell}_2 \quad \text{s.t.} \quad u_\ell \in
  \arg\min_{\red{u} \in \mathbb{C}^{\ell}} \norm{L_\ell \red{u}-
    t_\ell}.
\end{align}
$Q_{\ell,\ell+1}$ is neither formed nor applied (see~\eqref{QRfac}
and~\eqref{min7}), and the \QR factorization stops.  To obtain the
minimum-length solution, we still need to apply
$P_{\ell-2,\ell}P_{\ell-1,\ell}$ on the right of $R_\ell$ and
$\overline{V}_\ell$ in~\eqref{qlpRightRef} and~\eqref{eq:wvp}\red{,}
respectively.  If $b \in \range(A)$, then $L_\ell$ is nonsingular\red{,} and
the process in the previous subsection applies.  If $b \not\in
\range(A)$, the last row and column of $L_\ell$ are zero, \red{that is},
$L_\ell = \smat{L_{\ell-1} \\ 0 & 0}$ (see~\eqref{qlpeqn3a}), and we
need to define $u_\ell \equiv \smat{u_{\ell-1} \\ 0}$ and solve only
the last two equations of $L_{\ell-1} u_{\ell-1} = t_{\ell-1}$:
%, which is simply a $2 \times 2$ nonsingular  triangular system:
  \begin{equation}  \label{eq:Lut}
   %  L_k = \bmat{L_{k-1} \\ 0 & 0}  \text{,}\quad
   %  u_k = \bmat{u_{k-1} \\ 0}      \text{,}\quad
     \bmat{\gamma_{\ell-2}^{(6)}
        \\ \vartheta_{\ell-1}^{(2)} & \gamma_{\ell-1}^{(5)}
          }
     \bmat{\mu_{\ell-2}^{(3)}
        \\ \mu_{\ell-1}^{(2)}
          }
   = \bmat{\tau^{(2)}_{\ell-2}
        \\ \tau^{(2)}_{\ell-1}
       }.
  \end{equation}
\red{Recurrence}~\eqref{qlpeqnsol2} simplifies to $x_\ell =
x_{\ell-2}^{(2)} + w_{\ell-1}^{(3)} \mu_{\ell-1}^{(2)}$.  The following theorem
proves that \CSMINRESQLP yields $x^{\dagger}$ in this last iteration.

%\begin{proposition} \label{min-len-krylov} If $\hat{x}$ is the
%minimum-length solution of $Ax \approx b$ in a Krylov subspace, then
%$\hat{x} = x^\dagger$.  \end{proposition} \begin{proof} Let $\hat{x}
%= \arg\min\{\norm{x} \mid A\T Ax=A\T b, \; x \in
%\mathcal{K}_k(A,b)\}$.  First we want to show that $\hat{x} \in
%\range(A)$. Suppose $\hat{x}$ has a non-zero component in
%$\Null(A)$. Then $\hat{x} = x_r + x_n$ where $x_r \in \range(A)$ and
%$0\ne x_n\in \Null(A)$, and $\norm{\hat{x}}^2 = \norm{x_r}^2 +
%\norm{x_n}^2 > \norm{x_r}^2$ and $A\T A x_r = A\T b$. Thus $x_r$ is
%the minimum-length solution of $\norm{Ax-b}$ in the $k$-th Krylov
%subspace, which is a contradiction.  We know that $x^\dagger$ is
%unique and $x^\dagger = \arg\min\{\norm{x} \mid A\T Ax=A\T b, \; x
%\in \mathbb{R}^n\} = \arg\min\{\norm{x} \mid A\T Ax=A\T b, \; x \in
%\range(A)\}$. Since $\hat{x} \in \range(A)$, we must have $\hat{x} =
%x^\dagger$.  \end{proof}

\smallskip
\begin{theorem}   \label{theorem-MINRES-QLP}
  In \CSMINRESQLP, $x_\ell = x^\dagger$. %is the minimum-length
                                           %solution of $Ax \approx
                                           %b$.
\end{theorem}
\smallskip

\begin{proof}
%This can be seen from~\eqref{Lsubproblem}--\eqref{eq:Lut} with $k=\ell$.
  When $b \in \range(A)$, the proof is the same as that for
  Theorem~\ref{theorem-singular-compatible}.

  When $b \notin \range(A)$, for all $u = [u_{\ell-1}\T \ \,\mu_\ell]\T
  \in \mathbb{C}^\ell$ that solves~\eqref{Lsubproblem}, \CSMINRESQLP
  returns the minimum-length \LS solution $u_\ell = [u_{\ell-1}\T \ \,0]\T$
  by the construction in~\eqref{eq:Lut}.  For any
  $x\in\range(W_{\ell})= \range(\overline{V}_\ell)  
  \subseteq \range(\overline{A})=\range({A^*})$ by~\eqref{eq:wvp} and 
  $A\overline{V}_\ell=V_\ell T_\ell$,
\begin{align*}
  \norm{Ax-b} &= \norm{AW_\ell u - b} = \norm{A\conj{V}\!_\ell P_\ell u - b}
    = \norm{V_\ell T_\ell P_\ell u - \beta_1V_\ell e_1}
    =       \norm{T_\ell P_\ell u - \beta_1 e_1}
\\ &=\normm{Q_{\ell-1} T_\ell P_\ell u - \bmat{t_{\ell-1} \\ \phi_{\ell}}} =
   \normm{\bmat{L_{\ell-1} & 0 \\ 0 & 0} u -\bmat{t_{\ell-1} \\ \phi_{\ell}}}.
\end{align*}
Since $L_{\ell-1}$ is nonsingular, $ |\phi_{\ell}| = \min
\norm{Ax-b}$ can be achieved by $x_\ell = W_\ell u_\ell = W_{\ell-1}
u_{\ell-1}$ and $\norm{x_\ell} = \norm{W_{\ell-1} u_{\ell-1}} =
\norm{u_{\ell-1}}$. Thus $x_\ell$ is the minimum-length
\LS solution of $\norm{Ax-b}$,
\red{that is}, $x_\ell = \arg\min\{\norm{x} \mid A^*A x=A^* b, \; x \in
\range({A^*})\}$.  Likewise $y_\ell = P_\ell u_\ell$
is the minimum-length \LS solution of $\norm{T_\ell y - \beta_1 e_1}$\red{,} and
so $y_\ell \in \range(T^*_\ell)$, \red{that is,} $y_\ell = T^*_\ell z =
\conj{T}\!_\ell z$ for some $z$. Thus $x_\ell = \conj{V}\!_\ell y_\ell =
\conj{V}\!_\ell \conj{T}\!_\ell z = \conj{A} V_\ell z = A^* V_\ell z\in\!
\range(A^*)$.  We know that $x^\dagger = \arg\min\{\norm{x} \mid A^* A
x=A^* b, \; x \in \mathbb{C}^n\}$ is unique and $x^\dagger \in
\range(A^*)$.  Since $x_\ell \in \range(A^*)$, we must have $x_\ell =
x^\dagger$.
%, i.e., $x_\ell$ is the minimum-length LS solution of $\norm{Ax-b}$ in
%$\mathbb{R}^n$.
\end{proof}

%%%%%%%%%%%%%%%%%%%%%%%%%%%%%%%%%%%%%%%%%%%%%%%%%%%%%%%%%%
\section{Transferring CS-MINRES to \CSMINRES-QLP}
%: Transfer from \CSMINRES to \CSMINRESQLP}
\label{sec:transfer}
%%%%%%%%%%%%%%%%%%%%%%%%%%%%%%%%%%%%%%%%%%%%%%%%%%%%%%%%%%

\CSMINRES and \CSMINRESQLP behave  
similarly on well-conditioned systems. However, %\CSMINRES is
                                                %cheaper in terms of
                                                %memory and flops.
compared \red{with} \CSMINRES, \CSMINRESQLP requires one more vector of
storage, and each iteration needs 4 more axpy \red{operations} ($y \leftarrow \alpha x
+ y$) and 3 more vector \red{scalings} ($x \leftarrow \alpha x$).  It would
be a desirable feature to invoke \CSMINRESQLP from \CSMINRES only if
$A$ is ill-conditioned or singular.  The key idea is to transfer %from
\CSMINRES to \CSMINRESQLP at an iteration $k < \ell$ when $\underTk$
has full column rank and is still well-conditioned. At such an
iteration, the \CSMINRES point $x_k^M$ and \CSMINRESQLP point $x_k$
are the same, so from~\eqref{minresxk}, \eqref{qlpeqnsol2},
and~\eqref{Lsubproblem}: $ x_k^M = x_k \Longleftrightarrow D_k t_k =
W_k L_k^{-1} t_k$.  From \eqref{eq:rdeqv}, \eqref{qlpeqn3a},
and~\eqref{eq:wvp},
\begin{equation}  \label{transfereqn1}
  D_k L_k=(\conj{V}\!_kR_k^{-1})(R_kP_k)=\conj{V}\!_kP_k=W_k.
\end{equation}
The vertical arrow in Figure~\ref{fig:bases} represents this
process. In particular, we transfer only the last three \CSMINRES
basis vectors in $D_k$ to the last three \CSMINRESQLP basis vectors in
$W_k$:
\begin{equation}  \label{transfereqn2} 
     \bmat{w_{k-2} & w_{k-1} & w_k} 
     = \bmat{d_{k-2} & d_{k-1} & d_k} \bmat{\gamma_{k-2}^{(6)}
     \\ \vartheta_{k-1}^{(2)} & \gamma_{k-1}^{(5)}  &
     \\ \eta_{k}              & \vartheta_{k}       & \gamma_{k}^{(4)}
       }.
\end{equation}
%(Thus, we transfer the three
%\CSMINRES basis vectors $d_{k-2}, d_{k-1}, d_k$ to $w_{k-2}, w_{k-1},
%w_k$.)  

Furthermore, we need to generate the \CSMINRESQLP point
$\smash{x_{k-3}^{(2)}}$ in~\eqref{qlpeqnsol1} from the \CSMINRES point
$x_{k-1}^M$ by rearranging~\eqref{qlpeqnsol2}:
\begin{equation}  
  x_{k-3}^{(2)} = x_{k-1}^M - w_{k-2}^{(3)} \mu_{k-2}^{(2)}
    - w_{k-1}^{(2)} \mu_{k-1}.
\end{equation}  
Then the \CSMINRESQLP points $\smash{x_{k-2}^{(2)}}$ and $x_k$ can be
computed \red{by}~\eqref{qlpeqnsol1} and~\eqref{qlpeqnsol2}.

\red{From}~\eqref{transfereqn1} and~\eqref{transfereqn2}  we \red{clearly}
still need to do the right transformation $R_k P_k = L_k$ in the
\CSMINRES phase and keep the last $3 \times 3$ bottom right submatrix
of $L_k$ for each $k$ so that we are ready to transfer to \CSMINRESQLP
when necessary.  We then obtain a short recurrence for $\norm{x_k}$
(see Section~\ref{subsectsolnorm})\red{,} and for this computation we save
flops relative to the standalone \CSMINRES algorithm, which computes
$\norm{x_k}$ directly in the NRBE condition associated with
$\norm{r_k}$ in Table~\ref{tab-stopping-conditions}.

In the implementation of \CSMINRESQLP, the iterates transfer from
\CSMINRES to \CSMINRESQLP when an estimate of the condition number of
$T_k$ (see \eqref{cond2AQLP}) exceeds an input parameter
$\mathit{trancond}$.  Thus, $\mathit{trancond} > 1/\varepsilon$ leads
to \CSMINRES iterates throughout (that is, \CSMINRES standalone),
while $\mathit{trancond} = 1$ generates \CSMINRESQLP iterates from the
start (that is, \CSMINRESQLP standalone).

\begin{figure}    %%% Figure 5.1
\begin{center}
\vspace{3ex}
\includegraphics[width=1.8in]{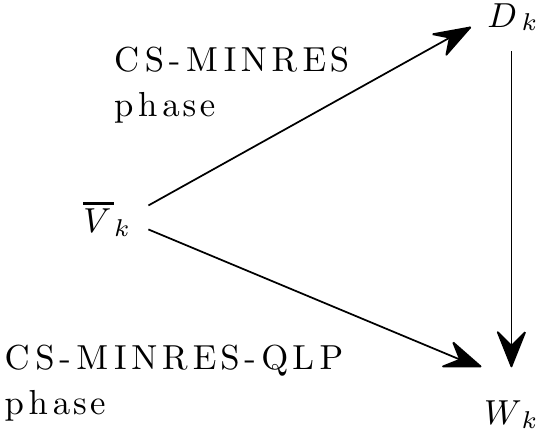}
\vspace{3ex}
\end{center}
\label{basis}
\caption{Changes of basis vectors within and between the two phases
  CS-MINRES and CS-MINRES-QLP; see equations~\eqref{eq:rdeqv},
  \eqref{eq:wvp}, and~\eqref{transfereqn2} for details.}
%\\ \tiny \green{Reproducible by
%\texttt{cs_csminres_Test.m},\texttt{cs_qmr_Test.m}, } %\vspace{-1ex}
\label{fig:bases}
\end{figure}

%%%%%%%%%%%%%%%%%%%%%%%%%%%%%%%%%%%%%%%%%%%%%%%%%%%%%%%%%%
\section{Preconditioned \CSMINRES and \CSMINRESQLP}  \label{sec:pcsminres}
%%%%%%%%%%%%%%%%%%%%%%%%%%%%%%%%%%%%%%%%%%%%%%%%%%%%%%%%%%

\begin{comment}
It is often asked: How can we construct a preconditioner for a linear
system solver so that the same problem is solved with fewer
iterations?  Previous work on preconditioning the symmetric solvers
\CG, \SYMMLQ, or \CSMINRES includes~\cite{RS02, N90, GMPS92, D98,
  FRSW98, NV98, RZ99, MGW00, H01, BL03, TPC04}.%SC: taken out GR01 and
added RS02

We have the same question for singular symmetric systems $Ax \approx
b$.
\end{comment}

Well-constructed two-sided preconditioners can preserve problem
symmetry and substantially reduce the number of iterations for
nonsingular problems. For singular compatible problems, we can still
solve the problems faster but generally obtain  \LS solutions
that are not of minimum length. \red{This} is not an issue due to 
algorithms but the way two-sided preconditioning is set up for
singular problems.  For incompatible systems (which are necessarily
singular), preconditioning alters the ``least squares'' norm.  To
avoid this difficulty\red{,} we could work with larger equivalent systems
that are compatible (see approaches
in~\cite[Section~7.3]{CPS11}), or we could apply a \red{right} 
preconditioner $M$ preferably such that $AM$ is complex symmetric so
that our algorithms are directly applicable. For example, if $M$ is
nonsingular (complex) symmetric and $AM$ is commutative, then $AMy
\approx b$ is a complex symmetric problem with $y \equiv M^{-1} x$.
This approach is efficient and straightforward. We devote the rest of
this section to \red{deriving} a two-sided preconditioning method.

%------------------------------------------------
%\subsection{Derivation}
%------------------------------------------------

We use a symmetric positive definite or a \emph{nonsingular} complex
symmetric preconditioner $M$.  For such an $M$, it is known that the
Cholesky factorization exists, \red{that is}, $M = CC^T$ for some lower
triangular matrix~$C$, which is real if $M$ is real, or complex if $M$
is complex.  We may employ commonly used construction techniques of
preconditioners \red{such as} diagonal preconditioning and incomplete Cholesky
factorization if the nonzero entries of $A$ are accessible.  It may
seem unnatural to use a symmetric positive definite preconditioner for
a complex symmetric problem. However, if available, its application
may be less expensive than a complex symmetric preconditioner.

We denote the square root of $M$ \red{by} \red{$M^{\myhalf}$}. It is known that a
complex symmetric root always exists for a nonsingular complex
symmetric $M$ even though it may not be unique;
see~\cite[Theorems~7.1, 7.2, and~7.3]{H08} or~\cite[Section
  6.4]{HJ91}.
%(say $C\equiv U\Sigma^{1/2}$ if an \SVD of $M$ is $U \Sigma U\T$) 
Preconditioned \CSMINRES (or \CSMINRESQLP) applies itself to the
equivalent system $\tilde{A} \tilde{x} = \tilde{b}$, where $\tilde{A}
\equiv M^{-\myhalf} A M^{-\myhalf}$, $\tilde{b} \equiv M^{-\myhalf}
b$, and $x = M^{-\myhalf}\tilde{x}$.  Implicitly, we are solving an
equivalent complex symmetric system $C\inv AC^{-T} y = C\inv b$, where
$C\T x = y$. In practice, we work with $M$ itself (solving the linear
system in~(\ref{pminresd4})). For analysis, we can assume $C =
M^{\myhalf}$ for convenience.
An effective preconditioner for \CSMINRES or \CSMINRESQLP is one such
that~$\tilde{A}$ has a more clustered eigenspectrum and \red{becomes} better
conditioned, and it is inexpensive to solve linear systems that
involve~$M$.

%------------------------------------------------
\subsection{Preconditioned \red{Saunders} process}
%------------------------------------------------

Let \red{$\conj{V}_k$} denote the \red{Saunders} vectors of the $k$th extended Krylov subspace
generated by $\tilde{A}$ and $\tilde{b}$.  With $v_0
= 0$ and $\beta_1 v_1 = \tilde{b}$, for $k=1,2,\ldots$ we define
\begin{equation}  \label{pminresd4}
   z_k = \beta_k M^{ \myhalf} v_k, \qquad
   q_k = \beta_k M^{-\myhalf} \conj{v}_k,
   \qquad \text{so that} \quad M \conj{q}_k =z_k.
\end{equation}
Then
\(
  \beta_k = \norm{\beta_k v_k}
    %  = \norm{M^{-\myhalf} \! z_k}
    %  = \norm{ z_k} _{M^{-1}}
    %  = \norm{ q_k }_M 
      = \sqrt{ q_k\T z_k },
\)
and the \red{Saunders} iteration is
\begin{align*}
   p_k &= \tilde{A} \conj{v}_k = M^{-\myhalf} \! A M^{-\myhalf} \conj{v}_k
                                     = M^{-\myhalf} \! A q_k / {\beta_k},
\\ \alpha_k           &= v_k^* p_k = q_k\T A q_k / {\beta_k^2},
\\ \beta_{k+1}v_{k+1} &= M^{-\myhalf}\! A M^{-\myhalf} \conj{v}_k
                            -\alpha_k v_k - \beta_k v_{k-1}.
\end{align*}
Multiplying the last equation by $M^{\myhalf}$\red{,} we get
\begin{align*}
   z_{k+1} = \beta_{k+1} M^{\myhalf} v_{k+1}
                & = A M^{-\myhalf} \conj{v}_k - \alpha_k M^{\myhalf} v_k
                                      - \beta_k  M^{\myhalf} v_{k-1}
\\         &= \frac{1}{\beta_k} A q_k - \frac{\alpha_k}{\beta_k} z_k
                                      - \frac{\beta_k}{\beta_{k-1}} z_{k-1}.
\end{align*}
The last expression involving consecutive $z_j$'s replaces the
three-term recurrence in $v_j$'s.  In addition, we need to solve a
linear system $M q_k = z_k$~\eqref{pminresd4} at each iteration.

%------------------------------------------------
\subsection{Preconditioned \CSMINRES}
%------------------------------------------------

From~\eqref{minresxk} and~\eqref{eq:rdeqv} we have the following
recurrence for the $k$th column of $D_k = \conj{V}\!_k R_k\inv$ and
$\tilde{x}_k$:
\[
  d_k = \bigl( \conj{v}_k - \delta_k^{(2)} d_{k-1}-\epsilon_k d_{k-2} \bigr)
           / \gamma_k^{(2)},
  \qquad
  \tilde{x}_k = \tilde{x}_{k-1} + \tau_k^{(2)} d_k.
\]
Multiplying the above two equations by $M^{-\myhalf}$ on the left and
defining $\tilde{d}_k = M^{-\myhalf}d_k$, we can update the solution
of our original problem by
\[
  \tilde{d}_k = \Bigl( \frac{1}{\beta_k} q_k
                - \delta_k^{(2)} \tilde{d}_{k-1} 
                - \epsilon_k \tilde{d}_{k-2} \Bigr)
                 \!\bigm/\! \gamma_k^{(2)},
  \qquad
  x_k = M^{-\myhalf} \tilde{x}_k
      = x_{k-1} + \tau_k ^{(2)} \tilde{d}_k.
\]
%We list the algorithm in \cite[Table~3.4]{C06}.

%------------------------------------------------ 
\subsection{Preconditioned \CSMINRESQLP} 
  \label{secPMINRESQLP}
%------------------------------------------------ 

A preconditioned \red{\CSMINRESQLP} can be derived similarly.  The
additional work is to apply right reflectors $P_k$ to $R_k$, and the
new subproblem bases are $W_k \equiv \conj{V}\!_k P_k$, with
$\tilde{x}_k = W_k u_k$. Multiplying the new basis and solution
estimate by $M^{-\myhalf}$ on the left, we obtain
\begin{align*}
   \widetilde{W}_k &\equiv M^{-\myhalf} W_k = M^{-\myhalf} \conj{V}\!_k P_k,
\\            x_k & = M^{-\myhalf} \tilde{x}_k
                    = M^{-\myhalf} W_k {u}_k
                    = \widetilde{W}_k {u}_k
                    = x_{k-2}^{(2)} +
                      \mu_{k-1}^{(2)} \tilde{w}_{k-1}^{(3)} +
                      \mu_k \tilde{w}_k^{(2)}.
\end{align*}

\section{Numerical experiments}  \label{sec:numerical}
%%%%%%%%%%%%%%%%%%%%%%%%%%%%%%%%%%%%%%%%%%%%%%%%%%%%%%%%%%

In this section we present computational results based on \red{the} \Matlab~7.12
\red{implementations} of \CSMINRESQLP and \SSMINRESQLP, which are made
available to the public as open-source software and \red{accord} with the
philosophy of reproducible computational research~\cite{C94, CD02}. 
The computations were performed in double precision on a Mac
OS X machine with a \red{2.7 GHz} Intel Core i7 and \red{16 GB} RAM.

\subsection{Complex symmetric problems} \label{sec:num:cs}
Although the SJSU Singular Matrix Database~\cite{FSJSU} currently
contains only one complex symmetric matrix (named \texttt{dwg961a}) and
only one skew symmetric matrix (\texttt{plsk1919}), it has a sizable set of
singular symmetric matrices, which can be handled by the associated
\Matlab toolbox SJsingular~\cite{FSJSU2}. We constructed multiple
singular complex symmetric systems of the form $H \equiv i A$, where
$A$ is symmetric and singular. \red{All} the eigenvalues of
$H$ \red{clearly} lie on the imaginary axis. For a compatible system, we simulated
$b = Hz$, where $z_i \sim i.i.d.\ U(0,1)$, \red{that is}, $z_i$ were
independent and identically distributed random variables whose values
were drawn from the standard uniform distribution with support
$[0,1]$. For a \LS problem, we generated a random $b$ with $b_{i} \sim
i.i.d.\ U(0,1)$\red{,} and it is almost always true that $b$ is \emph{not} in
the range of the test matrix. In \CSMINRESQLP, we set the parameters
$\mathtt{maxit} = 4n$, $\mathtt{tol} = \varepsilon$, and
$\mathtt{trancond} = 10^{-7}$ for the stopping conditions in
Table~\ref{tab-stopping-conditions} and the transfer process from
\CSMINRES (see Section~\ref{sec:transfer}).

%\QMR: 2  do not converge;  \CSMINRES  1

We compare the computed results of \CSMINRESQLP and \Matlab's \QMR \red{with}
solutions computed directly by the truncated \SVD (\TSVD)
of $H$ utilizing \red{\Matlab}'s function \texttt{pinv}.  For \TSVD we have
%\begin{align*}
$   x_t \equiv \sum_{\sigma_i > t \norm{H} \varepsilon}
               \frac{1}{\sigma_i} u_i u_i^* b,
  % \qquad
  % \norm{A}  =  \sigma_1,
  % \qquad
  % \kappa_t(A) = \frac{ \sigma_1}
  %             {\underset{\sigma_i > t \norm{A}\varepsilon} \min \sigma_i}
$
%\end{align*}
with parameter $t>0$.  Often $t$ is set to $1$, and sometimes to a
moderate number such as $10$ or $100$; it defines a cut-off point
relative to the largest singular value of $H$.  For example, if most
singular values are of order 1 and the rest are of order
$\norm{H}\varepsilon\approx10^{-16}$, we expect \TSVD to work better
when the small singular values are excluded, while \SVD (with $t=0$)
could return an exploding solution.

In Figure~\ref{cons50a} we present the results of 50 consistent
problems of the form $Hx=b$\red{. Given the computed TSVD solution $x^\dagger$,
the figure} plots the relative error norm \red{$\norm{\hat{x}-x^\dagger} / \norm{x^\dagger}$} of
approximate solution \red{$\hat{x}$} computed by \red{\QMR and \CSMINRESQLP}  with respect
to \TSVD solution against  
\red{$\kappa(H)\varepsilon$. (}It is known that an upper bound 
\red{on} the perturbation error of a singular linear
system involves the condition of the corresponding
matrix~\cite[Theorem~5.1]{SS}.\red{)} The diagonal dotted red line represents
the best results we could expect from any numerical method with double
precision.  We can see that both \red{\QMR and \CSMINRESQLP} did well
on all problems except for \red{two} in each case.  \CSMINRESQLP performed
slightly better \red{because} a few additional problems \red{solved by} \QMR attained relative
errors of less than $10^{-5}$.

\begin{figure}[ht]    %%% Figure 8.1
\includegraphics[width=4.7in]{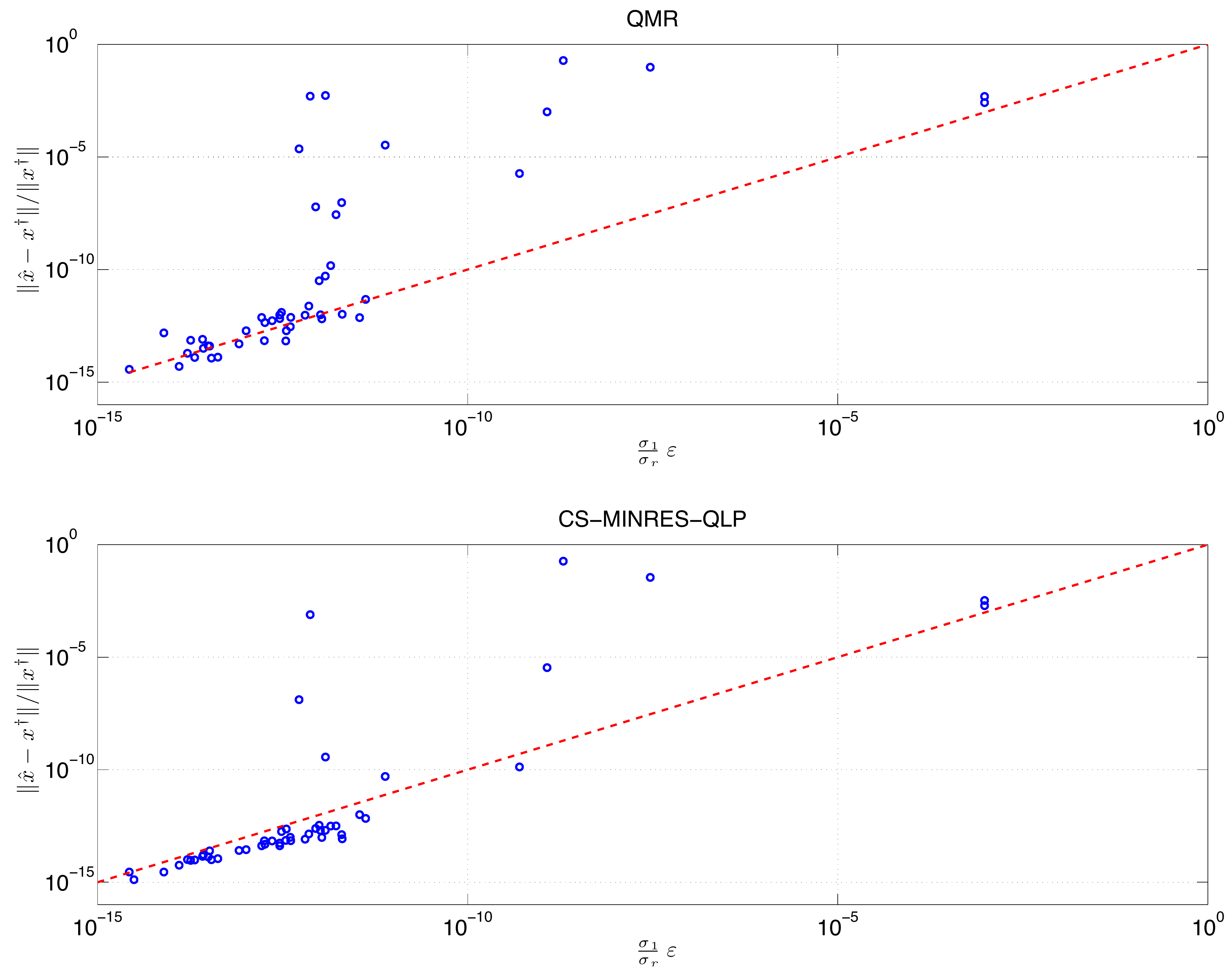}%{csminresTest2MethodsConsistent.pdf}
\label{cons50a}
\caption{50 consistent singular complex symmetric systems. This figure
  is reproducible by \texttt{C13Fig7\_1.m}.}
\end{figure}

Our second test set \red{involves} complex symmetric matrices that have \red{a} more 
wide-spread eigenspectrum than those in the first test set.  Let $A=V
\Lambda V\T$ be an eigenvalue decomposition of symmetric $A$ with
$|\lambda_1| \ge \cdots \ge |\lambda_n|$. For $i=1,\ldots,n$, we
define $d_i \equiv (2u_i-1) |\lambda_1|$, where $u_i \sim
i.i.d.\ U(0,1)$ if $\lambda_i \ne 0$, or $d_i \equiv 0$ otherwise.
Then the complex symmetric matrix $M \equiv V D V\T + i A$ has the
same (numerical) rank as $A$\red{,} and its eigenspectrum is bounded by a
ball of radius approximately equal to $|\lambda_1|$ on the complex
plane. In Figure~\ref{cons50b} we summarize the results of solving 50
such complex symmetric linear systems.   \CSMINRESQLP \red{clearly}
behaved as stably as it did with the first test set.  However, \QMR is
obviously more sensitive to the nonlinear spectrum: two problems did
not converge\red{,} and about ten additional problems converged to their
corresponding ${x}^{\dagger}$ with no more than four digits of
accuracy.

%QMR : \red{$>4$} do not converge; \quad \green{CS-MINRES}: \red{1}

Our third test set consists of linear \LS problems~\eqref{eqn4b}, in
which $A \equiv H$ in the upper plot of Figure~\ref{incons100} and $A
\equiv M$ in the lower plot.  In the case of $H$, \CSMINRESQLP did not
converge for two instances but agreed with the \TSVD solutions \red{to} five
or more digits for almost all other instances.  In the
case of $M$, \CSMINRESQLP did not converge for five instances but
agreed with the \TSVD solutions \red{to} five or more digits  
for almost all other instances. Thus the algorithm is to some
extent more sensitive to a nonlinear eigenspectrum in LS
problems\red{. This} is also expected by the perturbation result that an
upper bound of the relative error norm in a LS problem involves the
square of $\kappa(A)$~\cite[Theorem~5.2]{SS}. We did not run \QMR on
these test cases \red{because} the algorithm was not designed for \LS
problems.

{
%\begin{minipage}{1\textwidth}
\begin{figure}%[t]    %%% Figure 8.2
%\hspace{-8ex}
\includegraphics[width=4.5in]{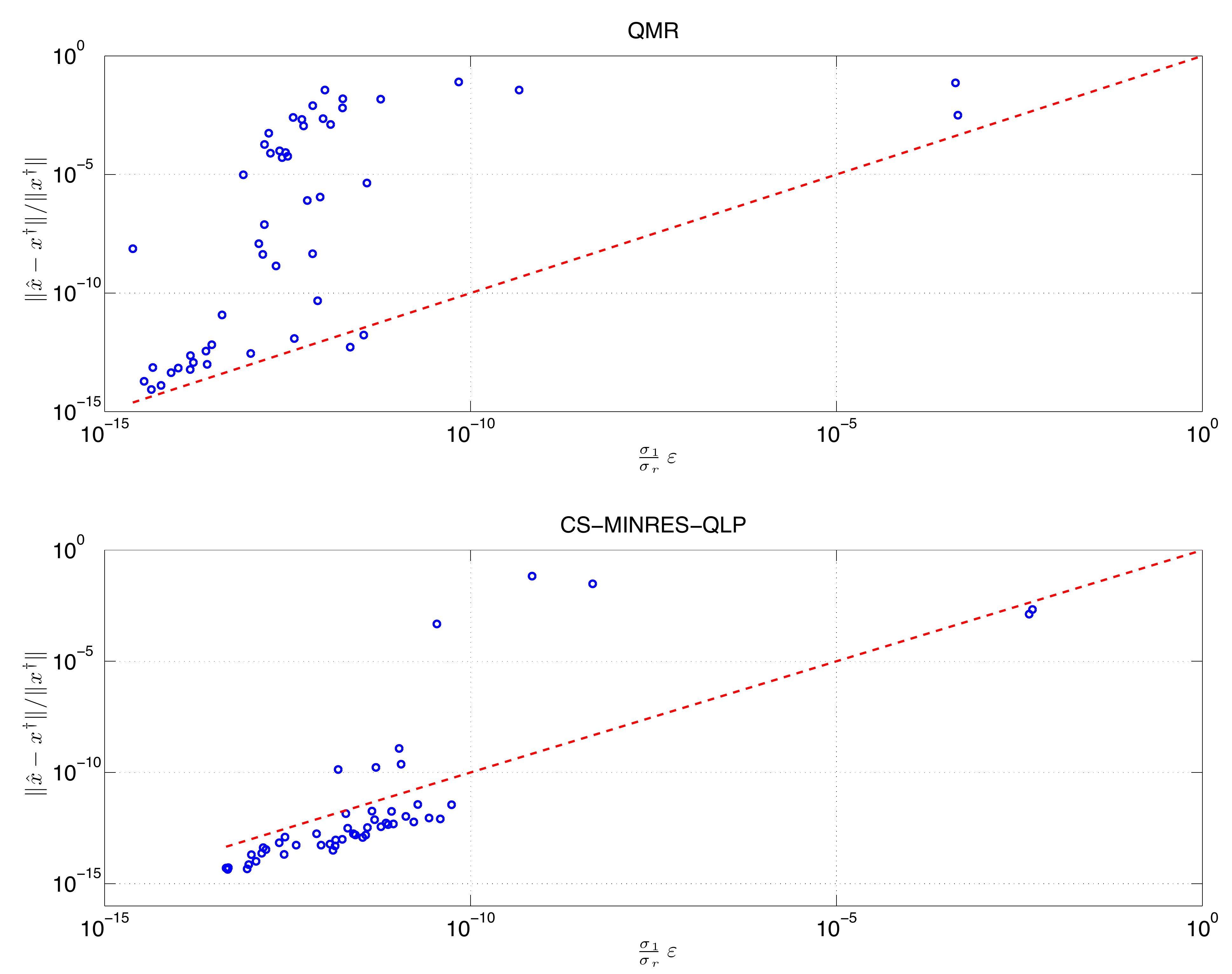}
\caption{50 consistent singular complex symmetric systems. This figure
  is reproducible by \texttt{C13Fig7\_2.m}.}
\label{cons50b}
%\\ \tiny \green{Reproducible by
%\texttt{cs_csminres_Test.m},\texttt{cs_qmr_Test.m}, } %\vspace{-1ex}
\end{figure}
 
\begin{figure}%[h]   %%% Figure 8.3 
%\hspace{-8ex}
\includegraphics[width=4.5in]{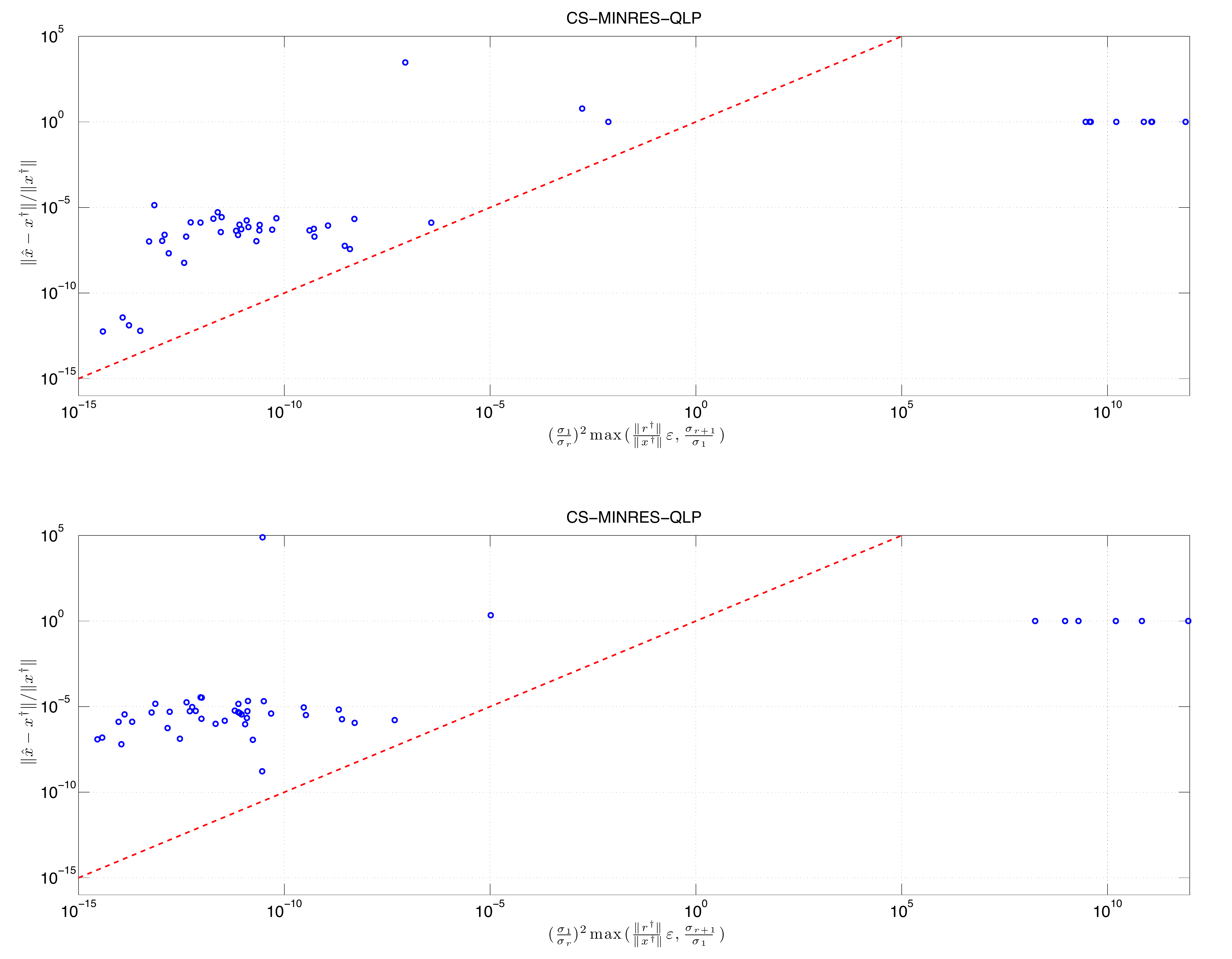}
\caption{100 \emph{inconsistent} singular complex symmetric
  systems. \red{We used matrices $H$ in the upper plot and $M$ in the lower plot, 
  where $H$ and $M$ are defined in Section~\ref{sec:num:cs}}. This figure is reproducible by \texttt{C13Fig7\_3.m}.}
\label{incons100}
%\\ \tiny \green{Reproducible by
%\texttt{cs_csminres_Test.m},\texttt{cs_qmr_Test.m}, } %\vspace{-1ex}
\end{figure}
%\end{minipage}
}

\subsection{Skew symmetric problems}
Our fourth test collection consists of 50 skew symmetric linear
systems and 50 singular skew symmetric \LS problems~\eqref{eqn4b}.
The matrices are constructed by $S = \mathtt{tril}(A) -
\mathtt{tril}(A)\T$, where $\mathtt{tril}$ extracts the lower
triangular part of a matrix.  In both cases---linear systems in the
upper subplot of Figure~\ref{100ss} and LS problems in the lower
subplot---\SSMINRESQLP did not converge for six instances but agreed
with the \TSVD solutions for more than ten digits of accuracy for
almost all other instances.

\begin{figure}%[h]   %%% Figure 7.4 
%\hspace{-8ex}
\includegraphics[width=4.5in]{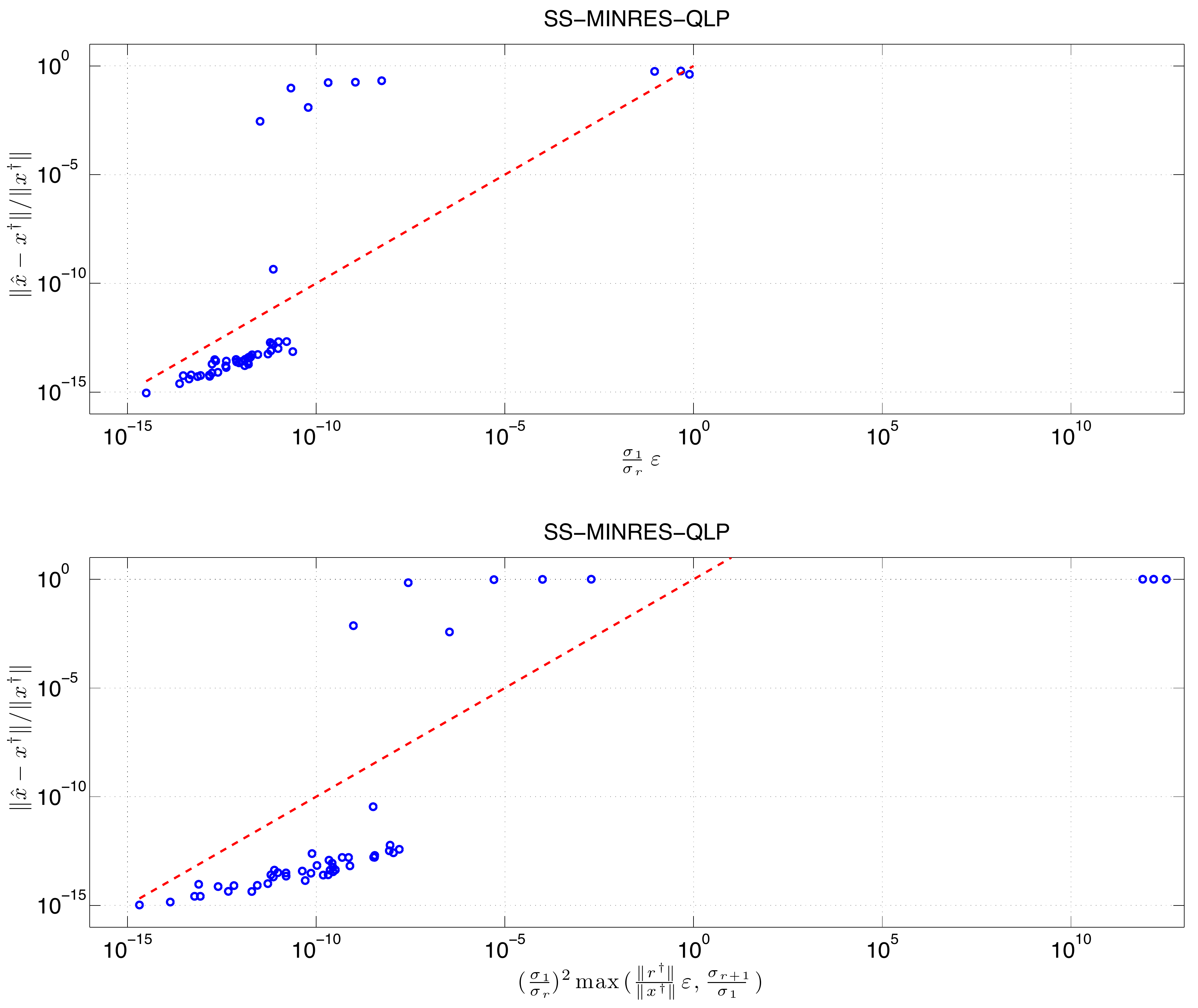}
\caption{100 singular skew symmetric systems.  Upper: 50 compatible
  linear systems. Lower: 50 LS problems. This figure is reproducible
  by \texttt{C13Fig7\_4.m}.}
\label{100ss}
%\\ \tiny \green{Reproducible by
%\texttt{cs_csminres_Test.m},\texttt{cs_qmr_Test.m}, } %\vspace{-1ex}
\end{figure}
%\end{minipage}

\subsection{Skew Hermitian problems}
We also have created a test collection of 50 skew
Hermitian linear systems and 50 skew Hermitian \LS
problems~\eqref{eqn4b}.  Each  skew Hermitian matrix is
constructed \red{as} $T = S + iB$, where $S$ is skew symmetric as defined in
the last test set, and $B \equiv A - \diag([a_{11},\ldots,a_{nn}])$\red{; in other words,}
$B$ is $A$ with the diagonal elements set to zero and is thus
symmetric.  We solve the problems using the original \MINRESQLP for
Hermitian problems by the transformation $(iT) x \approx ib$.  In the
case of linear systems in the upper subplot of Figure~\ref{100sh},
\SHMINRESQLP did not converge for \red{six instances}. For the other instances \SHMINRESQLP computed
approximate solutions that matched the \TSVD solutions for more than
ten digits of \red{accuracy.}  As for \red{the} LS
problems in the lower subplot of Figure~\ref{100sh}, only five
instances did not converge.

\begin{figure}%[h]   %%% Figure 7.4 
%\hspace{-8ex}
\includegraphics[width=4.5in]{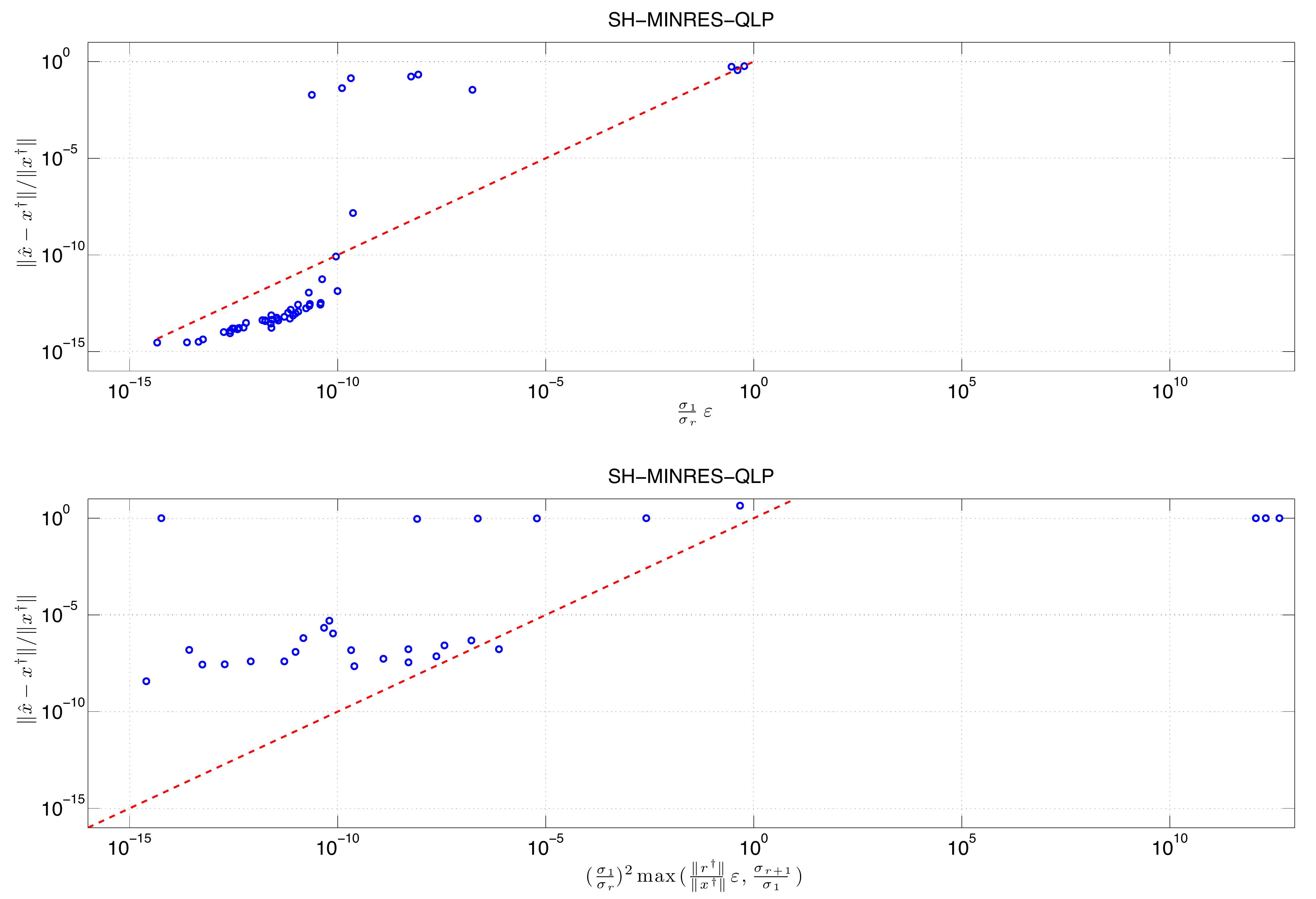}
\caption{100 singular skew Hermitian systems.  Upper: 50 compatible
  linear systems. Lower: 50 LS problems. This figure is reproducible
  by \texttt{C13Fig7\_5.m}.}
\label{100sh}
%\\ \tiny \green{Reproducible by
%\texttt{cs_csminres_Test.m},\texttt{cs_qmr_Test.m}, } %\vspace{-1ex}
\end{figure}
\section{Conclusions}  \label{sec:conclusions}
%%%%%%%%%%%%%%%%%%%%%%%%%%%%%%%%%%%%%%%%%%%%%%%%%%%%%%%%%%
\red{We take advantage of two Lanczos-like frameworks for square
matrices or linear operators with special symmetries. In particular, the framework
for complex-symmetric problems~\cite{BS99} is a special case of the Saunders-Simon-Yip process~\cite{SSY88} with the starting vectors chosen to be $b$ and $\bar{b}$; we name the complex-symmetric process the \textit{Saunders process} and the corresponding extended Krylov subspace  from the $k$th iteration  \textit{Saunders subspace} $S_k(A,b)$.} 

\CSMINRES constructs its $k$th solution estimate from the short recursion
$x_k = D_k t_k = x_{k-1} + \tau_k d_k$~\eqref{minresxk}, where
$n$ separate triangular systems $R_k\T D_k\T = V_k^*$ are solved to
obtain the $n$ elements of each direction $d_1, \ldots, d_k$.  (Only
$d_k$ is obtained during iteration $k$, but it has $n$ elements.)
In contrast, \CSMINRESQLP constructs $x_k$ using
orthogonal steps: $x_k = W_ku_k =x_{k-2}^{(2)} + w_{k-1}^{(3)} \mu_{k-1}^{(2)}
 + w_k^{(2)} \mu_k$; see
\eqref{qlpeqnsol1}--\eqref{qlpeqnsol2}.  Only one triangular system
$L_k u_k = t_k$ \eqref{Lsubproblem} is involved for each $k$.
Thus \CSMINRESQLP is numerically more stable than \CSMINRES.  
\red{The} additional work and storage are moderate, and 
efficiency is retained by transferring from \CSMINRES to  
\CSMINRESQLP   only when the estimated condition  of $A$
exceeds an input parameter value.

\TSVD is known to use rank-$k$ approximations to $A$ to find
approximate solutions to $\min\norm{Ax - b}$ that serve as a form of
\textit{regularization}.  It is fair to conclude from the results that
like other Krylov methods \CSMINRES have built-in regularization
features~\cite{HO93, HN96, KS99}.  Since \CSMINRESQLP monitors more
carefully and constructively the rank of $T_k$, which could be $k$ or
$k\!-\!1$, we may say that regularization is a stronger feature in
\CSMINRESQLP, as we have shown in our numerical examples.

Like \CSMINRES and \CSMINRESQLP, \SSMINRES and \SSMINRESQLP are
readily applicable to skew symmetric linear systems. \red{Similarly, we have
\SHMINRES and \SHMINRESQLP for skew Hermitian problems.} We summarize and
compare these methods in Appendix~\ref{sec:compare}.
%For both algorithms, we have derived recurrence relations
%for $\norm{r_k}$, $\norm{Ar_k}$ and $\norm{Ax_k}$ and applied them in
%stopping conditions for both singular and nonsingular problems.
\red{\CG and \SYMMLQ for problems with these special symmetries can be derived likewise.}
 
\begin{comment}  
It is important to develop robust techniques for estimating an a
priori bound for the solution norm since the \CSMINRES approximations
are not monotonic as is the case in \CG and \LSQR.

Ideally,
we would also like to determine a practical threshold associated with
the stopping condition $\gamma_k^{(4)} = 0$ in order to handle cases
when $\gamma_k^{(4)}$ is numerically small but not exactly zero.
These are topics for future research.
\end{comment}

%------------------------------------------------ 
\subsection*{\red{Software and reproducible research}}  
\label{sec:software}
%------------------------------------------------ 

\Matlab~7.12 and Fortran~90/95 implementations of \MINRES and
\MINRESQLP for symmetric, Hermitian, skew symmetric, skew Hermitian,
and complex symmetric linear systems with short-recurrence solution
and norm estimates as well as efficient stopping conditions are
available from the \MINRESQLP project website~\cite{MinresqlpMatlab}.

Following the philosophy of reliable reproducible computational research as
advocated in~\cite{C94, CD02, C13c}, for each figure and example in this
paper we mention either the source or the specific \Matlab command.
Our \Matlab scripts are available at~\cite{MinresqlpMatlab}.

%------------------------------------------------ 
\subsection*{\red{Acknowledgments}}  \label{sec:ack}
%------------------------------------------------ 
%
We thank our anonymous reviewers for their insightful suggestions for
improving this manuscript.  We are grateful \red{for} the feedback and
encouragement from Peter Benner, Jed Brown, \red{Xiao-Wen Chang, Youngsoo Choi}, Heike Fassbender, Gregory
Fasshauer, Ian Foster, \red{Pedro Freitas},  Roland Freund, Fred Hickernell, Ilse Ipsen,
Sven Leyffer, Lek-Heng Lim, \red{Lawrence Ma}, Sayan Mukherjee, Todd Munson, Nilima
Nigam, Christopher Paige, Xiaobai Sun, Paul Tupper, Stefan Wild,  
Jianlin Xia, \red{and Yuan Yao} during the development of this work. We appreciate the opportunities for presenting our algorithms in the 2012 Copper
Mountain Conference on Iterative Methods, the 2012 Structured
Numerical Linear and Multilinear Algebra Problems conference, and the
2013 SIAM Conference on Computational Science and Engineering. In
particular, we thank Gail Pieper, Michael Saunders, \red{and} Daniel Szyld for
their detailed comments, which resulted in a more polished
exposition. We thank Heike Fassbender, Roland Freund, and Michael
Saunders for helpful discussions on modified Krylov subspace methods
during Lothar Reichel's 60th birthday conference.

\clearpage
 
\appendix

%%%%%%%%%%%%%%%%%%%%%%%%%%%%%%%%%%%%%%%%%%%%%%%%%%%%%%%%%%
\section{Pseudocode of algorithms} \label{sec:pseudo}
%%%%%%%%%%%%%%%%%%%%%%%%%%%%%%%%%%%%%%%%%%%%%%%%%%%%%%%%%%

\begin{algo}{h}

  \Inputs{$A,b,\sigma,\operatorname{maxit}$\red{, where $A$ is complex symmetric, $\sigma$ is real or complex}}

  \smallskip
   $v_{0}=0$, \qquad $v_1=b$, \qquad $\beta_{1}=\left\Vert
  b\right\Vert$, \qquad $k=0$
  
%  \If{$\beta_1 > 0$}  
%  {$v_1 = b/\beta_1$ }
%  \Else{ STOP }  
  
  \While{$k\leq\operatorname{maxit}$}{
     \lIf{$\beta_{k+1} > 0$}    {$v_{k+1} \leftarrow v_{k+1}/\beta_{k+1}$ }
     \lElse{ STOP }  
     
     $k\leftarrow k+1$
      
      $p_{k}=A\conj{v}_{k}-\sigma \red{\conj{v}_{k}}$
      
      $\alpha_{k}=v_{k}^* p_{k}$
      
      $p_{k}\leftarrow p_{k}-\alpha_{k}v_{k}$

      $v_{k+1}=p_{k}-\beta_{k}v_{k-1}$
      
      $\beta_{k+1}= \norm{v_{k+1}}$ 
   }
  \Outputs{$V_\ell, T_\ell$}
  \caption{\red{Saunders process}.}
  \label{algo-cs-lanczos}
\end{algo}

\begin{algo}{h}

  \Inputs{$A,b,\operatorname{maxit}$, \red{where $A$ is skew symmetric}}

  \smallskip
   $v_{0}=0$, \qquad $v_1=b$, \qquad $\beta_{0} =0 $, \qquad
  $\beta_{1}=\left\Vert b\right\Vert$, \qquad $k=0$
  
%  \If{$\beta_1 > 0$}  
%  {$v_1 = b/\beta_1$ }
%  \Else{ STOP }  
  
  \While{$k\leq\operatorname{maxit}$}{
     \lIf{$\beta_{k+1} > 0$}    {$v_{k+1} \leftarrow v_{k+1}/\beta_{k+1}$ }
     \lElse{ STOP }  
     
     $k\leftarrow k+1$
      
      $p_{k}=Av_{k}- v_{k}$

      $v_{k+1}=p_{k}-\beta_{k}v_{k-1}$
      
      $\beta_{k+1}= \norm{v_{k+1}}$  
   }
  \Outputs{$V_\ell, T_\ell$}
  \caption{SS-Lanczos.}
  \label{algo-ss-lanczos}
\end{algo}

\begin{algo}{h}

  \Inputs{$A,b,\operatorname{maxit}$\red{, where $A$ is skew Hermitian}}

  \smallskip
   $v_{0}=0$, \qquad $v_1=\red{i}b$, \qquad $\beta_{1}=\left\Vert
  b\right\Vert$, \qquad $k=0$
  
%  \If{$\beta_1 > 0$}  
%  {$v_1 = b/\beta_1$ }
%  \Else{ STOP }  
  
  \While{$k\leq\operatorname{maxit}$}{
     \lIf{$\beta_{k+1} > 0$}    {$v_{k+1} \leftarrow v_{k+1}/\beta_{k+1}$ }
     \lElse{ STOP }  
     
     $k\leftarrow k+1$
      
      $p_{k}= i Av_{k}-i v_{k}$
      
      $\alpha_{k}=v_{k}^* p_{k}$
      
      $p_{k}\leftarrow p_{k}-\alpha_{k}v_{k}$

      $v_{k+1}=p_{k}-\beta_{k}v_{k-1}$
      
      $\beta_{k+1}= \norm{v_{k+1}}$  
   }
  \Outputs{$V_\ell, T_\ell$}
  \caption{SH-Lanczos.}
  \label{algo-sh-lanczos}
\end{algo}

\begin{algo}{tb}

  \Inputs{$a,b$}

  \smallskip

  \lIf{$\left\vert b \right\vert =0$} 
         {$c=1$, \qquad $s=0$, \qquad $r = a $\;}
  \ElseIf{$\left\vert a \right\vert=0$}
         {$c=0$, \qquad $s=1$, \qquad  $r=b$}
  \ElseIf{$\left\vert b \right\vert \ge \left\vert a \right\vert$}
         {$\tau=\left\vert a \right\vert /\left\vert b \right\vert$, 
         \qquad $c=1/\sqrt{1+\tau^{2}}$,
         \qquad $s=c \tau (\overline{b/a})$, \qquad $c \leftarrow c \tau$, \qquad $r=b/\overline{s}$}
  \ElseIf{$\left\vert a \right\vert > \left\vert b \right\vert$}
         {$\tau=\left\vert b \right\vert /\left\vert a \right\vert$, \qquad $c=1/\sqrt{1+\tau^{2}}$,
         \qquad $s=c  (\overline{b/a}) $, \qquad $r=a/c$}

  \Outputs{$c,s,r$}
  \caption{Algorithm SymOrtho.}
  \label{algo-symortho}
\end{algo}

\begin{algo}{p}

  \Inputs{$A, b, \sigma, M$\red{, where $A=A^T\in \mathbb{C}^{n\times n}$, $M= M\T\in \mathbb{C}^{n\times n}$ is nonsingular, $\sigma \in \mathbb{C}$ }}

  \smallskip

  $z_0 = 0$,
  \qquad $z_1 = b$,
  \qquad Solve $Mq_1 = z_1$,
  \qquad $\beta_1 = \sqrt{b\T q_1}$
  \tcc*[f]{Initialize}\;

  $w_0 = w_{-1} = 0$,
  \qquad $x_{-2} = x_{-1} = x_0 = 0$\;
  $c_{0,1} \!=\! c_{0,2} \!=\! c_{0,3} \!=\! -1$,
  \quad $s_{0,1} \!=\! s_{0,2} \!=\! s_{0,3}\!=\! 0$,
  \quad $\phi_0 \!=\! \beta_1$,
  \quad $\tau_0 \!=\! \omega_0 \!=\! \chi_{-2} \!=\! \chi_{-1} \!=\! \chi_0 \!=\! 0$\;

  $\delta_1 = \gamma_{-1} = \gamma_0
   = \eta_{-1} = \eta_0 = \eta_1
   = \vartheta_{-1} = \vartheta_0 = \vartheta_1
   = \mu_{-1} = \mu_0 = 0$,
  \quad $\mathcal{A} = 0, \quad  \kappa = 1$\;

  $k=0$\;

  \BlankLine

  \While{no stopping condition is satisfied}{
    $k\leftarrow k+1$\;

    $p_k = Aq_k - \sigma q_k$,
    \qquad $\alpha_k = \frac{1}{\beta_k^2}q_k\T p_k$
    \tcc*[f]{Preconditioned Lanczos}\;

    $z_{k+1} = \frac{1}{\beta_k} p_k - \frac{\alpha_k}{\beta_k}z_k
             - \frac{\beta_k}{\beta_{k-1}} z_{k-1}$\;

    Solve $Mq_{k+1} = z_{k+1}$,
    \qquad $\beta_{k+1} = \sqrt{q_{k+1}\T z_{k+1}}$\;

    \lIf{$k = 1$}{$\rho_k = \norm{[\alpha_k \ \; \beta_{k+1}]}$}
    \lElse{$\rho_k = \norm{[\beta_k \ \; \alpha_k \ \; \beta_{k+1}]}$}\;

    $\delta_k^{(2)} = c_{k-1,1} \delta_k+s_{k-1,1} \alpha_k$
    \tcc*[f]{Previous left reflection\dots}\;

    $\gamma_k = s_{k-1,1}\delta_k -c_{k-1,1} \alpha_k$
    \tcc*[f]{on middle two entries of $\underTk e_k$\dots}\;

    ${\epsilon}_{k+1} = s_{k-1,1} \beta_{k+1}$
    \tcc*[f]{produces first two entries in $\underTkp e_{k+1}$}\;

    $\delta_{k+1} = -c_{k-1,1} \beta_{k+1}$\;

    $c_{k1}, s_{k1}, \gamma_k^{(2)}
      \leftarrow \SymOrtho(\gamma_k, \beta_{k+1})$
    \tcc*[f]{Current left reflection}\;

    $c_{k2}, s_{k2}, \gamma_{k-2}^{(6)}
      \leftarrow \SymOrtho(\gamma_{k-2}^{(5)}, \epsilon_k)$
    \tcc*[f]{First right reflection}\;

    $\delta_k^{(3)} = s_{k2}\vartheta_{k-1} - c_{k2} \delta_k^{(2)}$,
    \qquad $\gamma_k^{(3)} = -c_{k2}\gamma_k^{(2)}$,
    \qquad $\eta_k = s_{k2} \gamma_k^{(2)}$\;

    $\vartheta_{k-1}^{(2)} = c_{k2} \vartheta_{k-1} + s_{k2} \delta_k^{(2)}$\;

    $c_{k3}, s_{k3}, \gamma_{k-1}^{(5)}
      \leftarrow \SymOrtho(\gamma_{k-1}^{(4)}, \delta_k^{(3)})$
    \tcc*[f]{Second right reflection\dots}\;

    $\vartheta_k = s_{k3}\gamma_k^{(3)}$,
    \qquad $\gamma_k^{(4)} = -c_{k3}\gamma_k^{(3)}$
    \tcc*[f]{to zero out $\delta_k^{(3)}$}\;

    $\tau_k^{(2)} = c_{k1} \tau_{k}$
    \tcc*[f]{Last element of $t_k$}\;

    $\phi_k =\phi_{k-1} |s_{k1}|$, \quad $\psi_{k-1} = \phi_{k-1}
         \norm{[\smash{\gamma_k \ \; \delta_{k+1}}]}$
    \tcc*[f]{Update $\norm{r_k}$, $\norm{Ar_{k-1}}$}\;

    \lIf{$k=1$}{$\gamma_{\min}=\gamma_{1}$}
    \lElse{$\gamma_{\min} \leftarrow \min{ \{ \gamma_{\min},
           \gamma_{k-2}^{(6)}, \gamma_{k-1}^{(5)}, | \gamma_k^{(4)}|  \}}$}\;

    $\mathcal{A}^{(k)} = \max{ \{\mathcal{A}^{(k-1)}, \rho_k,
           \gamma_{k-2}^{(6)}, \gamma_{k-1}^{(5)}, |\gamma_k^{(4)}| \}}$
    \tcc*[f]{Update $\norm{A}$}\;

    $\omega_k = \norm{[\smash{\omega_{k-1} \ \; \tau_k^{(2)}}]}$,
    \qquad $\kappa \leftarrow \mathcal{A}^{(k)} / \gamma_{\min}$
    \tcc*[f]{Update $\norm{A x_k}$, $\kappa(A)$}\;

    $w_k = -(c_{k2} / \beta_k) q_k + s_{k2} w_{k-2}^{(3)}$
    \tcc*[f]{Update $w_{k-2}$, $w_{k-1}$, $w_k$}\;

    $w_{k-2}^{(4)} = (s_{k2} / \beta_k) q_k + c_{k2} w_{k-2}^{(3)}$

    \lIf{$k>2$}
        {$w_k^{(2)} = s_{k3} w_{k-1}^{(2)} - c_{k3} w_k$,
         \qquad $w_{k-1}^{(3)} = c_{k3} w_{k-1}^{(2)} + s_{k3} w_k$}\;

    \lIf{$k>2$}
        {$\mu_{k-2}^{(3)} = (\tau_{k-2}^{(2)} - \eta_{k-2} \mu_{k-4}^{(4)}
                                        - \vartheta_{k-2} \mu_{k-3}^{(3)})
                            / \gamma_{k-2}^{(6)}$}
    \tcc*[f]{Update $\mu_{k-2}$}\;

    \lIf{$k>1$}
        {$\mu_{k-1}^{(2)} =(\tau_{k-1}^{(2)} - \eta_{k-1} \mu_{k-3}^{(3)} -
         \vartheta_{k-1}^{(2)} \mu_{k-2}^{(3)}) / \gamma_{k-1}^{(5)}$}
    \tcc*[f]{Update $\mu_{k-1}$}\;

    \lIf{$\gamma_k^{(4)} \neq 0$}
        {$\mu_k = (\tau_k^{(2)} - \eta_k \mu_{k-2}^{(3)}
         - \vartheta_k \mu_{k-1}^{(2)}) / \gamma_k^{(4)}$}
    \lElse{$\mu_k = 0$}
    \tcc*[f]{Compute $\mu_k$}\;

    $x_{k-2}^{(2)}  = x_{k-3}^{(2)}  + \mu_{k-2}^{(3)} w_{k-2}^{(3)}$
    \tcc*[f]{Update $x_{k-2}$}\;

    $x_k = x_{k-2}^{(2)}  + \mu_{k-1}^{(2)} w_{k-1}^{(3)}
                       + \mu_k w_k^{(2)}$
    \tcc*[f]{Compute $x_{k}$}\;

    $\chi_{k-2}^{(2)} = \norm{[\smash{\chi_{k-3}^{(2)}  \ \; \mu_{k-2}^{(3)}}]}$
    \tcc*[f]{Update $\norm{x_{k-2}}$}\;

    $\chi_k = \norm{[\smash{\chi_{k-2}^{(2)}  \ \; \mu_{k-1}^{(2)}
                                       \ \; \mu_k}]}$
    \tcc*[f]{Compute $\norm{x_k}$}\;
  }

  \BlankLine

  $x = x_k$,
  \quad $\phi = \phi_k$,
  \quad $\psi = \phi_k \norm{[\smash{\gamma_{k+1}
                                      \ \; \delta_{k+2}}]}$,
  \quad $\chi = \chi_k$,
  \quad $\mathcal{A} = \mathcal{A}^{(k)}$,
  \quad $\omega = \omega_k$\;

  \Outputs{$x, \phi ,\psi, \chi, \mathcal{A}, \kappa, \omega$}\;

  \tcc*[f]{$c,s \leftarrow \SymOrtho(a,b)$ \rm is a stable form for computing
           $r = \sqrt{a^2+b^2}$,  $c = \frac{a}{r}$,  $s = \frac{b}{r}$}\;

  \caption{Preconditioned CS-MINRES-QLP to solve $(A - \sigma I)x \approx b$.}
  \label{pcs-minresqlpalgo}
\end{algo}

%%%%%%%%%%%%%%%%%%%%%%%%%%%%%%%%%%%%%%%%%%%%%%%%%%%%%%%%%%
\section{Stopping conditions and norm estimates}  \label{sec:QLPstop}
%%%%%%%%%%%%%%%%%%%%%%%%%%%%%%%%%%%%%%%%%%%%%%%%%%%%%%%%%%

This section derives several short-recurrence norm estimates
in~\red{\MINRES and \MINRESQLP for complex symmetric and skew Hermitian systems}.  
As before, we assume exact arithmetic
throughout, so that $V_k$ and $Q_k$ are orthonormal.
Table~\ref{tab-stopping-conditions} summarizes how these norm
estimates are used to formulate six groups of concerted stopping
conditions.  The second NRBE test is specifically designed for \LS
problems, which have the properties $r\ne 0$ but $A^*r=0$; it is
inspired by Stewart~\cite{Stewart-Rice-1977} and stops the algorithm
when $\norm{A^* r_k}$ is small relative to its upper bound $\norm{A}
\norm{r_k}$.

\begin{table}[h!]    %%% Table 6.1
\caption{Stopping conditions in CS-MINRES and SH-MINRES.  NRBE means
  normwise relative backward error, and $\mathit{tol}$,
  $\mathit{maxit}$, $\mathit{maxcond}$ and $\mathit{maxxnorm}$ are
  input parameters.  All norms and $\kappa(A)$ are estimated by
  CS-MINRES and SH-MINRES.}
\label{tab-stopping-conditions}
\begin{center}
  \begin{tabular}{|l|l|l|}
    \hline
     Lanczos & NRBE
              & Regularization attempts
  \\ \hline $ $ & &
  \\[-8pt]
     $\beta_{k+1} \le n\norm{A} \varepsilon$
             & $\norm{r_k} / \left(\norm{A}\norm{x_k}+\norm{b}\right)
                   \le \max(\mathit{tol},\varepsilon)$
              & $\kappa(A) \geq \max(\mathit{maxcond},\frac{1}{\varepsilon})$
  \\[4pt]
     $k = \mathit{maxit}$
             & $\norm{\conj{A}r_k} / \left(\norm{A} \norm{r_k} \right)
                   \le \max(\mathit{tol},\varepsilon)$
             & $\norm{x_k} \geq  \mathit{maxxnorm}$  
  \\[3pt]   \hline 
        \CSMINRES & Degenerate cases
                  & Erroneous input
  \\ \hline $ $ & &
  \\[-8pt]
   $ \left| \gamma_k^{(4)} \right| < \varepsilon $
              & $\beta_1=0\quad \Rightarrow \quad x^\dagger=0 $
              & $ y^*(Az) \ne \pm z^*(\conj{A}y)$
  \\[4pt]
               & $\beta_2=0  \quad \Rightarrow  \quad x^\dagger = \overline{b}/\alpha_1$ 
               & $ \qquad \Rightarrow \qquad A\ne \pm A\T$
  \\ \hline
  \end{tabular}
\end{center}
\end{table}

%------------------------------------------------ 
\subsection{Residual and residual norm}
%------------------------------------------------ 

First we derive recurrence relations for $r_k$ and its norm
$\norm{r_k} = |\phi_k|$. The results are true for \CSMINRES and
\CSMINRESQLP.
%In practice, $\rank(L_\ell)$ will be a reliable measure of $\rank(T_\ell)$.

%The following lemma says, among other things, that the intermediate
%$r_k$'s in \CSMINRES are \emph{not} orthogonal to
%$\mathcal{K}_k(A,b)$.

\smallskip
\begin{lemma}%[$r_k$ and $\norm{r_k}$ and monotonicity of
             %$\norm{r_k}$]
  \label{minresqlp_r} Without loss of generality, let $x_0=0$. \red{We have the results below.} 
%\[
%  \begin{array}{@{}l@{\quad}l@{\quad}l@{}}
%     \text{If } k < \ell, \text{ then } \rank(L_k)=k,
%        & r_k    = s_k^2 r_{k-1}\! -\! \phi_k c_k v_{k+1}, & \phi_k = \phi_{k-1} s_k;
%  \\ \text{If } k=\ell\ \&\ \rank(L_\ell)=\ell,       \text{ then}
%        & r_\ell = 0;                                   &
%  \\ \text{If } k=\ell\ \&\ \rank(L_\ell)=\ell\!-\!1, \text{ then}
%        & r_\ell = r_{\ell-1} \ne 0,
%        & \norm{r_\ell}  = \phi_{\ell-1};
%\end{array}
%\]
%where $\ \norm{r_k}=\phi_k$ when $k < \ell$.
%It follows that $\norm{r_k} \le \norm{r_{k-1}}$.
\begin{enumerate} 
\item $r_0=b$ and $\norm{r_0} = \phi_0 = \beta_1$.
\item For $k=1,\ldots, \ell-1$,  %$\rank(L_k) = k$,
         %$r_k = |s_k|^2 r_{k-1} \!-\! \phi_k c_k v_{k+1}$, and 
         $\norm{r_k} = |\phi_k| = |\phi_{k-1}| |s_k| \ge |\phi_{k-1}| > 0$. 
         Thus $\norm{r_k}$ is monotonically decreasing.
\item At the last iteration $ \ell$,
\begin{enumerate} 
\item If $\rank(L_\ell) = \ell$, then $\norm{r_\ell} = \phi_\ell=0$.
\vspace{1ex}\item If $\rank(L_\ell) = \ell\!-\!1$, then 
%$r_\ell = r_{\ell-1} \ne 0$, and 
$\norm{r_\ell} = |\phi_{\ell-1}| > 0$.
\end{enumerate} 
\end{enumerate}
\end{lemma}

\smallskip
\begin{proof}
\begin{enumerate}
\item Obvious.
\item 
  If $k < \ell$, %$R_k$ is nonsingular and $R_k$ is nonsingular.
  \red{from}~\eqref{eqn:rk}--\eqref{minresxk} with $R_ky_k=t_k$ we have
\begin{align}
   r_k %&= b -A x_k = V_{k+1} ( \beta_1 e_1 - \underTk y_k ) \nonumber
   &= V_{k+1} Q_k^* \pmat{ \bmat{t_k \\ \phi_k} - \bmat{R_k \\ 0}
     y_k} = \phi_k V_{k+1} Q_k^* e_{k+1}.  \label{rk7}
\end{align}
  We have $\norm{r_k} = |\phi_k| = |\phi_{k-1}||s_k| > 0$;
  see~\eqref{QRfac2}--\eqref{eq:normrk}.
 
\begin{comment}  
Also from $r_k = \phi_k V_{k+1} Q_k^T e_{k+1}$ \eqref{rk7} we have
\begin{align}
    r_k  &= \phi_k \bmat{V_k & v_{k+1}}
                   \bmat{Q_{k-1}^* & \\ & 1}
                   \hbox{\footnotesize $
                   \bmat{I_{k-1} & & \\ & c_k & s_k\\ & \conj{s}_k & -c_k}
                   \bmat{0 \\ 0 \\ 1}
                   $}
            \quad \mbox{by~\eqref{QRfac}},    \nonumber
\\[-3pt] \displaybreak[0]
         &= \phi_k \bmat{V_k & v_{k+1}}
                   \bmat{Q_{k-1}^* & \\ & 1}
                   \bmat{s_k e_k\\ -c_k}
           =\phi_k \bmat{V_k & v_{k+1}}
                   \bmat{s_kQ_{k-1}^* e_k\\ -c_k}
   \nonumber %\label{rk7b}
\\[-3pt] \displaybreak[0]
         &= \phi_k s_k V_k Q_{k-1}^* e_k - \phi_k c_k v_{k+1}
          = \tau_{k} |s_k|^2 V_k Q_{k-1}^* e_k - \phi_k c_k v_{k+1}
   \nonumber
\\ \displaybreak[0]
         &= |s_k|^2 r_{k-1} - \phi_k c_k v_{k+1} \text{ by~\eqref{rk7}}.
   \nonumber %\label{rk8}  
\end{align}
 \end{comment}
% If $k = \ell$, from Theorem~\ref{theorem-MINRES-QLP},
% %Remark~\ref{rem:ns}, $b\in \mathcal{R}(A)$ if and only if $T_k$ is
% nonsingular.  But $\rank(L_k)=k$ implies $R_k$ and $T_k$ are
% nonsingular, and so $b\in \mathcal{R}(A)$ and $r_k=0$.  If
% $\rank(L_\ell)=\ell\!-\!1$,
\item If $T_\ell$ is nonsingular, $r_\ell = 0$.  Otherwise
  $Q_{\ell-1,\ell}$ has made the last row of $R_\ell$ zero, so the
  last row and column of $L_\ell$ are zero; see \eqref{eq:Lut}.  Thus
  $r_\ell = r_{\ell-1}\ne 0$.
  %; see Remark~\ref{rem:stop2}.
\end{enumerate}
\end{proof}

%------------------------------------------------ 
\subsection{Norm of $A^*r_k$}
%------------------------------------------------
 
For incompatible systems, $r_k$ will never be zero.  However, all \LS
solutions satisfy $A^* Ax = A^*b$, so that $A^*r = 0$.  We therefore
need a stopping condition based on the size of $\norm{A^*
  r_k}=\psi_k$.  We present efficient recurrence relations for
$\norm{A^* r_k}$ in the following \red{lemma}. We also show that $A^*r_k$ is
orthogonal to $\mathcal{K}_k(A,\conj{b})$.

\smallskip
\begin{lemma}[$A^*r_k$ and $\psi_k \equiv \norm{A^*r_k}$ for \CSMINRES]
  \label{minreslemma2}
%If $\beta_{k+1}>0$, then
%\[
%  \begin{array}{@{}l@{\quad}l@{\quad}l@{}}
%     \rank(L_k)=k,
%     &   Ar_k = \norm{r_k} (\gamma_{k+1}v_{k+1}
%                          + \delta_{k+2}v_{k+2}),
%     & \psi_k = \norm{r_k} \norm{[\gamma_{k+1} \ \; \delta_{k+2}]},
%  \end{array}
%\]
%where $\delta_{k+2}v_{k+2}$ is zero if $\beta_{k+2}=0$.
%\[
%  \begin{array}{@{}l@{\quad}l@{\quad}l@{}}
%     \text{If } \beta_{k+1}=0\ \&\ \rank(L_k)=k,
%     &   Ar_k = 0,
%     & \psi_k = 0.
%  \\ \text{If } \beta_{k+1}=0\ \&\  \rank(L_k)=k\!-\!1,
%     &   Ar_k = Ar_{k-1},
%     & \psi_k = \psi_{k-1}.
%  \end{array}
%\]
\begin{enumerate}
\vspace{1ex}\item If $k < \ell$, then $\rank(L_k) = k$, $\conj{A}r_k =
\norm{r_k} (\conj{\gamma}_{k+1}\conj{v}_{k+1} +
\delta_{k+2}\conj{v}_{k+2})$ and $\psi_k = \norm{r_k}
\norm{[\gamma_{k+1} \ \; \delta_{k+2}]}$, where $\delta_{k+2}=0$ if $k
= \ell\!-\!1$.
\vspace{1ex}\item At the last iteration  $ \ell$, 
\begin{enumerate}
\vspace{1ex}\item If $\rank(L_\ell) = \ell$, then $Ar_\ell=0$ and
$\psi_\ell = 0$.
\vspace{1ex}\item If $\rank(L_\ell) = \ell\!-\!1$, then
         $Ar_\ell = Ar_{\ell-1} = 0$, and
         $\psi_\ell  = \psi_{\ell-1} = 0$.
\end{enumerate}
\end{enumerate}
\end{lemma}

\smallskip
\begin{proof}
%These follow from Lemmas~\ref{minreslemma2} and \ref{minresqlp_r}.
Case~2 follows directly from Lemma~\ref{minresqlp_r}. We prove the
first case here.  For $k < \ell$, $R_k$ is nonsingular.
From~\eqref{eqn:rk}--\eqref{minresxk} with $R_ky_k=t_k$ we have
\begin{align}
%   r_k %&= b -A x_k = V_{k+1} ( \beta_1 e_1 - \underTk y_k )     \nonumber
%       &= V_{k+1} Q_k^T \pmat{ \bmat{t_k \\ \phi_k} - \bmat{R_k \\ 0}  y_k}
%          = \phi_k V_{k+1} Q_k^T e_{k+1},                       % \label{rk7}
%\\
  \conj{A}r_k &= \phi_k \conj{V}\!_{k+2} \underline{\conj{T}\!_{k+1}} Q_k^* e_{k+1},  
  \text{ by } (\ref{rk7})              \nonumber
\\ Q_k \underTkp\T
       &= Q_k \bmat{T_{k+1} & \beta_{k+2} e_{k+1}}
        = Q_k \bmat{     T_k           & \beta_{k+1} e_k  & 0
                 \\ \beta_{k+1} e_k^T  & \alpha_{k+1}     & \beta_{k+2}},
                                                              \nonumber
\\ e_{k+1}^T Q_k \underTkp\T & =
                     \bmat{0 & \gamma_{k+1} & \delta_{k+2}},    \nonumber
\end{align}
by~\eqref{min7}.
% , where $A\conj{V}_{k+1} = V_{k+1} T_{k+1}$ and 
We take
$\delta_{k+2}=0$ if $k = \ell-1$, so
\begin{align*}
   \conj{A}r_k &= \tau_{k+1} \conj{V}\!_{k+2} 
           \bmat{0 & \gamma_{k+1} & \delta_{k+2}}^*
         = \tau_{k+1} \left( \conj{\gamma}_{k+1} \conj{v}_{k+1} + \delta_{k+2} 
           \conj{v}_{k+2}
                 \right),
\\ \psi_k^2 &\equiv \norm{\conj{A}r_k}^2 = \norm{r_k}^2
      \left( [\gamma_{k+1}]^2 + [\delta_{k+2}]^2 \right). 
\end{align*}
\red{The result} follows.
\end{proof}

Typically $\norm{\conj{A}r_k}$ is not monotonic, while clearly
$\norm{r_k}$ is monotonically decreasing.  In the singular system $A
=U \Sigma U\T$, let $U = \bmat{U_1 \!&\! U_2}$, where the singular
vectors $U_1$ correspond to nonzero singular values.  Then $P_A \equiv
U_1 U_1^*$ and $P^{\perp}_A \equiv U_2 U_2^*$ are orthogonal
projectors~\cite{TB} onto the range and nullspace of $A$.  For general
linear \LS problems, Chang et al.\ \cite{CPP09} characterize the
dynamics of $\norm{r_k}$ and $\norm{A^* r_k}$ in three phases defined
in terms of the ratios among $\norm{r_k}$, $\norm{P_A r_k}$\red{,} and
$\norm{P^{\perp}_A r_k}$, and propose two new stopping criteria for
iterative solvers.  The expositions in~\cite{AG08, JT10} show that
these estimates are cheaply computable in \CGLS and
\LSQR~\cite{PS82a,PS82b}. These results are likely applicable to
\CSMINRES.

%------------------------------------------------ 
\subsection{Matrix norms}
%------------------------------------------------ 

From the Lanczos process,
$\norm{A} \ge \norm{V_{k+1}^*A\conj{V}\!_k} = \norm{\underTk}$.
Define
\begin{equation}  \label{normA2}
  \mathcal{A}^{(0)} \equiv 0, \quad \mathcal{A}^{(k)}
    \equiv \max_{j=1,\ldots,k} \left\{ \norm{\underTj e_j} \right\}
    = \max\left\{\mathcal{A}^{(k-1)},\norm{\underTk e_k} \right\}
      \text{ for } k \ge 1.
\end{equation}
Then $\norm{A} \geq \norm{\underTk} \geq \mathcal{A}^{(k)}$.
Clearly, $\mathcal{A}^{(k)}$ is monotonically increasing and is thus
an improving estimate for $\norm{A}$ as $k$ increases.  By the
property of \QLP decomposition in~\eqref{qlpeqn1a} and~\eqref{qlpRightRef},
we could easily extend~\eqref{normA2} to include the largest diagonal
of $L_k$:
\begin{equation}  \label{normA2b}
  \mathcal{A}^{(0)} \equiv 0, \quad
  \mathcal{A}^{(k)} \equiv \max\{\mathcal{A}^{(k-1)},       \,
                                   \norm{ \underline {T_k}e_k}, \,
                                   \gamma_{k-2}^{(6)},          \,
                                   \gamma_{k-1}^{(5)},          \,
                                  |\gamma_k^{(4)}| \} \text { for } k\ge1,
\end{equation}
\begin{comment}
Some other schemes inspired by Larsen \cite[Section~A.6.1]{L98},
Higham~\cite{H} and Chen and Demmel~\cite{CD00} follow.  For the
latter scheme, we use an implementation by Kaustuv~\cite{K08} for
estimating the norms of the rows of $A$.
\begin{enumerate}
\item \cite{L98} $\norm{T_k}_1 \ge \norm{T_k} $
\item \cite{L98} $\sqrt{\norm{\underTk^T \underTk}_1} \ge \norm{T_k}  $
\item \cite{L98} $\norm{T_j} \le \norm{T_k}$  for small $j = 5$ or $20$
\item \cite{H} \Matlab function {\small NORMEST}$(A)$, which is
  based on the power method
\item \cite{CD00}
$\max_i \norm{h_i}/\sqrt{m}$, where
$h_i^T$ is the $i$th row of $AZ$,
each column of $Z \in \mathbb{R}^{n \times m}$ is a random vector of $\pm1$'s,
and $m$ is a small integer (e.g., $m=10$).
\end{enumerate}

Figure~\ref{AnormQLP} plots estimates of $\norm{A}$ for 12 matrices
from the Florida sparse matrix collection~\cite{UFSMC} whose sizes $n$
vary from 25 to 3002.  In particular, scheme 3 above with $j = 20$
gives significantly more accurate estimates than other schemes for the
12 matrices we tried.  However, the choice of $j$ is not always clear
and the scheme adds a little to the cost of \CSMINRES.  Hence we
propose incorporating it into \CSMINRES (or other Lanzcos-based
iterative methods) if very accurate $\norm{A}$ is needed.  Otherwise
\end{comment} 
which uses quantities readily available from \CSMINRES and gives
satisfactory, if not extremely accurate, estimates for the order of
$\norm{A}$.

%------------------------------------------------ 
\subsection{Matrix condition numbers}
%------------------------------------------------ 

We again apply the property of the \QLP decomposition in
\eqref{qlpeqn1a} and~\eqref{qlpRightRef} to estimate
$\kappa(\underTk)$, which is a lower bound for $\kappa(A)$:
\begin{align}
  \gamma_{\min} &\leftarrow \min\{\gamma_1, \gamma_2^{(2)}\}, \quad
  \gamma_{\min}  \leftarrow \min\{\gamma_{\min}, \gamma_{k-2}^{(6)}, \,
                               \gamma_{k-1}^{(5)}, \,  |\gamma_k^{(4)}| \}
                            \text{ for } k \ge 3, \nonumber
\\ \kappa^{(0)} &\equiv 1, \quad
   \kappa^{(k)} \equiv \max\left\{ \kappa^{(k-1)},
                                   \frac{\mathcal{A}^{(k)}}{\gamma_{\min}}
                          \right\}  \text{ for } k \ge 1.
   \label{cond2AQLP}
\end{align}

%------------------------------------------------ 
\subsection{Solution norms} \label{subsectsolnorm}
%------------------------------------------------ 

For \CSMINRESQLP, we derive a recurrence relation for $\norm{x_k}$
whose cost is as low as computing the norm of a $3$- or \red{$4$-}vector. 
This recurrence relation is not applicable to \CSMINRES
standalone.

Since $\norm{x_k} = \norm{\conj{V}\!_k P_k u_k} = \norm{u_k}$, we can
estimate $\norm{x_k}$ by computing $\chi_k \equiv \norm{u_k}$.
However, the last two elements of $u_k$ change in $u_{k+1}$ (and a new
element $\mu_{k+1}$ is added).  We therefore maintain $\chi_{k-2}$ by
updating it and then using it according to
\[
   \chi_{k-2}^{(2)} =
   \norm{[\chi_{k-3}^{(2)} \ \; \mu_{k-2}^{(3)}]},
   \quad
   \chi_{k}
   = \norm{[\chi_{k-2}^{(2)}  \ \; \mu_{k-1}^{(2)} \ \; \mu_k]}\red{;}
  % \quad \text{}.
\]
cf.~\eqref{qlpeqnsol1} and~\eqref{qlpeqnsol2}. Thus $\chi_{k-2}^{(2)}$ increases monotonically but we cannot
guarantee that $\norm{x_k}$ and its recurred estimate $\chi_k$ are
increasing, and indeed they are not in some examples. But the trend
for $\chi_k$ is generally increasing, and $\chi_{k}^{(2)}$ is
theoretically a better estimate than $\chi_{k}$ for $\norm{x_k}$.
%(see Figure~\ref{davis1177}). 
In \LS problems, when $\gamma_k^{(4)}$ is small enough in
magnitude, it also means $\norm{x_k} = \norm{y_k} = \norm{u_k}$ is
large---and when this quantity is larger than $\mathit{maxxnorm}$, it
usually means that we should do only a partial update of $x_k =
x_{k-2}^{(2)} + w_{k-1}^{(3)} \mu_{k-1}^{(2)}$. \red{If it} still exceeds
$\mathit{maxxnorm}$ in length, then we  do no update, \red{namely},
$x_{k} = x_{k-2}^{(2)}$.

%------------------------------------------------ 
\subsection{Projection norms}
%------------------------------------------------ 

In applications requiring nullvectors~\cite{C06}, $Ax_k$  is
useful.  Other times, the projection of the right-hand side $b$ onto
$\mathcal{K}_k(A,\conj{b})$ is required~\cite{S97}.  For the
recurrence relations of $Ax_k$ and its norm, we have
\begin{align*}
   Ax_k &= A\conj{V}\!_k y_k = V_{k+1} \underTk y_k
                    = V_{k+1} Q_k^*  \bmat{R_k \\ 0} y_k
                    = V_{k+1} Q_k^*  \bmat{t_k \\ 0},
\\ \omega_k^2 &\equiv \norm{Ax_k}^2 
    = \norm{t_k}^2
    = \norm{t_{k-1}}^2 +  (\tau_k^{(2)} )^2
    = \omega_{k-1}^2   +  (\tau_k^{(2)} )^2 
    = \norm{\smat{\omega_{k-1} & \tau_k^{(2)}}}^2.
\end{align*}
Thus $\{\omega_k\}$ is monotonic.

\begin{comment}
%%%%%%%%%%%%%%%%%%%%%%%%%%%%%%%%%%%%%%%%%%%%%%%%%%%%%%%%%%
\section{Proof that
   {\rm $T_\ell$ is nonsingular iff $b \in \range(A)$, and 
 $\rank(T_\ell) = \ell-1$ if $b \notin \range(A)$.
   (section \ref{sec:Lanproperties})}}
\label{appendixA}
%%%%%%%%%%%%%%%%%%%%%%%%%%%%%%%%%%%%%%%%%%%%%%%%%%%%%%%%%%
  
  We use $A \overline{V}_\ell = V_\ell T_\ell$ twice.  First, if
  $T_\ell$ is nonsingular, we can solve $T_\ell y_\ell = \beta_1 e_1$
  and then $A \overline{V}_\ell y_\ell = V_\ell T_\ell y_\ell = V_\ell
  \beta_1 e_1 = b$.  Conversely, if $b \in \range(A)$ then
  $\range(\overline{V}_\ell) \subseteq \range(A)$.  Suppose $T_\ell$
  is singular. Then there exists $z \ne 0$ such that $V_\ell T_\ell z
  = A \overline{V}_\ell z=0$. That is, $0 \ne \overline{V}_\ell z \in
  \Null(A)$.  But this is impossible because $\overline{V}_\ell z \in
  \range(A)$ and $\Null(A) \cap \range(\overline{V}_\ell) = 0$.  Thus,
  $T_\ell$ must be nonsingular.

We have shown if $b \notin \range(A)$, $T_\ell =
\bmat{\underline{T_{\ell-1}}&
      \begin{smallmatrix}\beta_\ell e_{\ell-1}
                         \\ \alpha_\ell
      \end{smallmatrix}}$
is singular, and therefore $\ell > \rank(T_\ell) \ge
\rank(\underline{T_{\ell-1}}) = \ell-1$ since $\rank(\underline{T_k})
= k$ for all $k<\ell$. Therefore, $\rank(T_\ell) = \ell-1$.  \endproof
 \end{comment}

%%%%%%%%%%%%%%%%%%%%%%%%%%%%%%%%%%%%%%%%%%%%%%%%%%%%%%%%%%
\section{Comparison of Lanczos-based solvers}  \label{sec:compare}
%%%%%%%%%%%%%%%%%%%%%%%%%%%%%%%%%%%%%%%%%%%%%%%%%%%%%%%%%%
 
We compare our new solvers with \CG, \SYMMLQ, \MINRES, and \MINRESQLP
in Tables~\ref{tableQLPsubproblems}--\ref{tableQLPbasis} in terms of
subproblem definitions, basis, solution estimates, flops, and memory.
A careful implementation of \SYMMLQ computes $x_k$ in
$\range(V_{k+1})$; see~\cite[Section~2.2.2]{C06} for a proof.  All
solvers need storage for $v_k$, $v_{k+1}$, $x_k$, and a product $p_k =
Av_k$ \text{ or } $A\conj{v}_k$ each iteration.  Some additional
work-vectors are needed for each method (e.g., $d_{k-1}$ and $d_k$ for
\MINRES or \CSMINRES, giving 7 work-vectors in total). \red{We
note} that even for Hermitian and skew Hermitian problems $Ax=b$,
the subproblems of \CG, \SYMMLQ, \MINRES, and \MINRESQLP 
are real.
\begin{table}[ht!]    %%% Table 5.1
  \caption{Subproblems defining $x_k$ for  CG, SYMMLQ, 
    MINRES, MINRES-QLP, CS-MINRES, CS-MINRES-QLP, \red{SS-MINRES, SH-MINRES-QLP,}  SH-MINRES, and
    SH-MINRES-QLP.}
  \label{tableQLPsubproblems}
  \centering
  \begin{tabular}{|l|l|l|l|}
   \hline
     Method  &  Subproblem  &  Factorization  &  Estimate of $x_k$
  \\ \hline \tablestrut
     cgLanczos            & $T_k y_k = \beta_1 e_1$ & Cholesky: & $x_k = V_k y_k$
  \\ \cite{HS52,PS75,SOL} &                         & $T_k = L_k D_k L_k\T$
                                                    & $\quad \in \red{\mathcal{K}_k(A,b)}$
  \\[0pt] \hline \tablestrut
     \SYMMLQ  & $y_{k+1} = \arg\min_{y \in \mathbb{R}^{k+1}} \norm{y}$
                        & \LQ:                     & $x_k = V_{k+1} y_{k+1}$
  \\ \cite{PS75,C06}      & \quad s.t.~$\underline{T_k\T} y = \beta_1 e_1$
                          & $\underTk\T  Q_k\T  = \bmat{L_k & \!\!\!0}$
                          & \quad $\in \red{\mathcal{K}_{k+1}(A,b)}$
  \\[4pt] \hline \tablestrut
    \MINRES\cite{PS75},\cite{C06}-\cite{CS12b}      & ${\displaystyle y_k =
                          \arg\min_{y\in\mathbb{R}^k}
                          \norm{\underTk y - \beta_1 e_1}}$
                                 & \QR:     & $x_k = V_k y_k$
  \\[-14pt]\red{\SSMINRES}  &      & $Q_k\underTk = \raisebox{4pt}{$\bmat{R_k \\ 0}$}$
                                           & $\quad\in\red{\mathcal{K}_k(A,b)}$
  \\[2pt] \hline \tablestrut
     \MINRESQLP\!\! \cite{C06}-\cite{CS12b}     & $y_k = \arg\min_{y\in\mathbb{R}^k} \norm{y}$
                                 & {\QLP:}  & $x_k =V_k y_k$
  \\[-10pt] \red{\SSMINRESQLP}  & s.t.~$y \in \arg\min \norm{\underTk y - \beta_1 e_1}$\!\!
                                 & $Q_k\underTk P_k = \raisebox{4pt}{$\bmat{L_k\\0}$}$
                                         & $\quad \in \red{\mathcal{K}_k(A,b)}$                                     
  \\[2pt] \hline \tablestrut
    \SHMINRES             & ${\displaystyle y_k =
                          \arg\min_{y\in\mathbb{R}^k}
                          \norm{\underTk y - \beta_1 e_1}}$
                                 & \QR:     & $x_k =  \red{V}\!_k y_k$
   \\[-10pt]                &      & $Q_k\underTk = \raisebox{4pt}{$\bmat{R_k \\ 0}$}$
                                           & $\quad\in\red{\mathcal{K}_k(iA,ib)}$     
  \\[2pt] \hline \tablestrut
     \SHMINRESQLP\!\!       & $y_k = \arg\min_{y\in\mathbb{R}^k} \norm{y}$
                                 & {\QLP:}  & $x_k = \red{V}\!_k y_k$
  \\[-4pt]                  & s.t.~$y \in \arg\min \norm{\underTk y - \beta_1 e_1}$\!\!
                                 & $Q_k\underTk P_k = \raisebox{4pt}{$\bmat{L_k\\0}$}$
                                         & $\quad \in \red{\mathcal{K}_k(iA,ib)}$                                                                                
  \\[2pt] \hline \tablestrut
    \CSMINRES             & ${\displaystyle y_k =
                          \arg\min_{y\in\mathbb{C}^k}
                          \norm{\underTk y - \beta_1 e_1}}$
                                 & \QR:     & $x_k = \conj{V}\!_k y_k$
   \\[-10pt]              &      & $Q_k\underTk = \raisebox{4pt}{$\bmat{R_k \\ 0}$}$
                                           & $\quad\in\red{\mathcal{S}_k(A,b)}$                        
  \\[2pt] \hline \tablestrut
     \CSMINRESQLP\!\!       & $y_k = \arg\min_{y\in\mathbb{C}^k} \norm{y}$
                                 & {\QLP:}  & $x_k =\conj{V}\!_k y_k$
  \\[-4pt]                  & s.t.~$y \in \arg\min \norm{\underTk y - \beta_1 e_1}$\!\!
                                 & $Q_k\underTk P_k = \raisebox{4pt}{$\bmat{L_k\\0}$}$
                                         & $\quad \in\red{\mathcal{S}_k(A,b)}$
  %
%   \\[2pt] \hline \tablestrut
%   \SSMINRES             & ${\displaystyle y_k =
%                          \arg\min_{y\in\mathbb{R}^k}
%                          \norm{\underTk y - \beta_1 e_1}}$
%                                 & \QR:     & $x_k =  {V}\!_k y_k$
%   \\[-10pt]                &      & $Q_k\underTk = \raisebox{4pt}{$\bmat{R_k \\ 0}$}$
%                                           & $\quad\in\red{\mathcal{K}_k(A,b)}$    
%  \\[2pt] \hline \tablestrut
%     \SSMINRESQLP\!\!       & $y_k = \arg\min_{y\in\mathbb{R}^k} \norm{y}$
%                                 & {\QLP:}  & $x_k = {V}\!_k y_k$
%  \\[-4pt]                  & s.t.~$y \in \arg\min \norm{\underTk y - \beta_1 e_1}$\!\!
%                                 & $Q_k\underTk P_k = \raisebox{4pt}{$\bmat{L_k\\0}$}$
%                                         & $\quad \in \red{\mathcal{K}_k(A,b)}$                                                                                
 
  \\[0pt] \hline
  \end{tabular}

\vspace*{0.4in}

  \caption{Bases, subproblem solutions, storage, and work for each method.}
  \label{tableQLPbasis}
  \centering
  \begin{tabular}{|l|l|l|l|c|}
     \hline
     Method & New \red{Basis} & $ \quad\quad\quad z_k, t_k, u_k $ & $x_k$ \red{Estimate}
                        & \!\!vecs \ flops\!\!
  \\ \hline \tablestrut
     cgLanczos & $W_k \equiv V_k L_k^{-T}$ & $ L_k D_k z_k =\beta_1 e_1$
               & $x_k \!=\! W_k z_k$
               & 5 \ \ $\ 8n$
  \\ \hline  \rule[0ex]{0ex}{4ex}%
     \SYMMLQ   & $W_k\equiv V_{k+1} Q_k^T \bmat{I_k \\ 0}$\!\! & $L_k z_k \!=\!\beta_1 e_1$
               & $x_k \!=\! W_k z_k$
               & 6 \ \ $\ 9n$
  \\[2pt] \hline \tablestrut
     \MINRES   & $D_k \equiv V_k R_k^{-1}$
               & $\ \ \ \ t_k \!=\!\beta_1 \mystrut\bmat{I_k & \!\!\!0} Q_k e_1$\!\!
               & $x_k \!=\! D_k t_k $
               & 7 \ \ $\ 9n$
   \\[-6pt]\red{\SSMINRES} & & & &
   \\[-2pt]\red{\SHMINRES} & & & &
  \\ \hline \tablestrut
     \MINRESQLP\!\! & $W_k \equiv V_k P_k$
               & $L_k u_k \!=\!\beta_1 \mystrut\bmat{I_k & \!\!\!0} Q_k e_1$\!\!
               & $x_k \!=\! W_k u_k$
               & 8 \ \ $14n$
  \\[-6pt]\red{\SSMINRESQLP} & & & &
  \\[-2pt]\red{\SHMINRESQLP} & & & &
  \\[2pt] \hline \tablestrut
     \CSMINRES & $D_k \equiv \conj{V}\!_k R_k^{-1}$
               & $\ \ \ \ t_k \!=\!\beta_1 \mystrut\bmat{I_k & \!\!\!0} Q_k e_1$\!\!
               & $x_k \!=\! D_k t_k $
               & 7 \ \ $\ 9n$
 % \\[-2pt]\SHMINRES & & & &
  \\[2pt] \hline \tablestrut
     \CSMINRESQLP\!\! & $W_k \equiv \conj{V}\!_k P_k$
               & $L_k u_k \!=\!\beta_1 \mystrut\bmat{I_k & \!\!\!0} Q_k e_1$\!\!
               & $x_k \!=\! W_k u_k$
               & 8 \ \ $14n$
 % \\[-2pt]\SHMINRESQLP & & & &
  \\[2pt] \hline
\end{tabular}
\end{table}

\pagebreak
 
\frenchspacing

\bibliographystyle{abbrv}
\bibliography{../../../../LaTeX_files/CSC_refs6}

\vfill

{\small The submitted manuscript has been created by the University of
  Chicago as Operator of Argonne National Laboratory (``Argonne'')
  under Contract DE-AC02-06CH11357 with the U.S.\ Department of Energy.
  The U.S.\ Government retains for itself, and others acting on its
  behalf, a paid-up, nonexclusive, irrevocable worldwide license in
  said article to reproduce, prepare derivative works, distribute
  copies to the public, and perform publicly and display publicly, by
  or on behalf of the Government.}

\end{document}

%% file: CsminresMacros.tex
\newcommand{\pmat}[1]{{\renewcommand{\arraystretch}{1.1}%
   \begin{pmatrix}#1\end{pmatrix}}}  % Need \usepackage{amsmath}
\newcommand{\bmat}[1]{{\renewcommand{\arraystretch}{1.1}%
   \begin{bmatrix}#1\end{bmatrix}}}
\newcommand{\smat}[1]{{\renewcommand{\arraystretch}{1.1}%
   \left[\begin{smallmatrix}#1\end{smallmatrix}\right]}}

\ifx   \innerprod\undefined%
   \def\innerprod(#1,#2){\langle#1,#2\rangle} % Must be called as \innerproduct(A,B)
\fi

\newcommand{\Newcommand}[2]%
   {\ifx#1\undefined \newcommand{#1}{#2} \else \renewcommand{#1}{#2} \fi}
\ifx\mod\undefined
  \newcommand{\mod}[1]{|#1|}
\else
  \renewcommand{\mod}[1]{|#1|}
\fi
\Newcommand{\Re} {\mathbb{R}}           % Need \usepackage{amssymb}

\providecommand{\diag} {\mathop{\mathrm{diag}}}

\providecommand{\rank} {\mathop{\mathrm{rank}}}

\providecommand{\CSMINRES} {{\small CS-MINRES}\xspace}
\providecommand{\CSMINRESQLP} {{\small CS-MINRES-QLP}\xspace}
\providecommand{\SHMINRES} {{\small SH-MINRES}\xspace}
\providecommand{\SSMINRES} {{\small SS-MINRES}\xspace}
\providecommand{\SSMINRESQLP} {{\small SS-MINRES-QLP}\xspace}
\providecommand{\SHMINRESQLP} {{\small SH-MINRES-QLP}\xspace}

\providecommand{\CG}        {{\small CG}\xspace}
\providecommand{\CGLS}      {{\small CGLS}\xspace}

\providecommand{\MINRESQLP} {{\small MINRES-QLP}\xspace}

\providecommand{\LSQR}      {{\small LSQR}\xspace}
\providecommand{\QMR} {{\small QMR}\xspace}
\providecommand{\BCG} {{\small BiCG}\xspace}

\newcommand{\SymOrtho}{\text{SymOrtho}\xspace}

\providecommand{\diag} {\mathop{\operator@font{diag}}}
\providecommand{\tril} {\mathop{\operator@font{tril}}}
\newcommand{\be}{\begin{enumerate}}
\newcommand{\ee}{\end{enumerate}}

\newcommand{\T}{^T\!}

\newcommand{\conj}[1]{\overline{#1}}

\newcommand{\Cpp}{C\raise3pt\hbox{\tiny++}}

\newcommand{\cond}{\mathrm{cond}}

\newcommand{\Null}{\mathop{\mathrm{null}}}

\newcommand{\range}{\mathop{\mathrm{range}}}

\newcommand{\Span}{\mbox{\rm span}}

\newcommand{\etal}{et al.}  %%% No italics!!  Also, must say \etal\

\newcommand{\inv}{^{-1}}

\newcommand{\normm}[1]{\biggl\|#1\biggr\|}
\newcommand{\norm}[1]{\|#1\|}
\newcommand{\spose}[1]{\hbox to 0pt{#1\hss}}

\providecommand{\text}[1]{\hbox{\quad#1\quad}}

\newcommand{\nthinsp}{\mskip -2   mu}

\Newcommand{\R}{_{\scriptscriptstyle R}}

\newcommand{\superstar}{^{\raise 0.5pt\hbox{$\nthinsp *$}}}
\newcommand{\SUPERSTAR}{^{\raise 0.5pt\hbox{$*$}}}
\newcommand{\lamstarT }{\lambda^{\raise 0.5pt\hbox{$\nthinsp *$}T}}

\newcommand{\rhat}{\skew3\widehat r}

\newcommand{\xhat}{\skew{2.8}\widehat x}

\providecommand{\Matlab}{{\sc Matlab}\xspace}

\providecommand{\GMRES}{{\small GMRES}\xspace}
\providecommand{\LSQR}{{\small LSQR}\xspace}
\providecommand{\MINRES}{{\small MINRES}\xspace}
\providecommand{\SYMMLQ}{{\small SYMMLQ}\xspace}
\providecommand{\BICGSTAB}{{\small BiCGstab}\xspace}

\providecommand{\LS}{{\small LS}\xspace} % least-squares
\providecommand{\SVD}{{\small SVD}\xspace}
\providecommand{\TSVD}{{\small TSVD}\xspace}
\providecommand{\QR}{{\small QR}\xspace}
\providecommand{\LQ}{{\small LQ}\xspace}
\providecommand{\QLP}{{\small QLP}\xspace}
\providecommand{\ULV}{{\small ULV}\xspace}

%% file: CSMINRES25.bbl
\begin{thebibliography}{10}

\bibitem{AG08}
M.~Arioli and S.~Gratton.
\newblock Least-squares problems, normal equations, and stopping criteria for
  the conjugate gradient method.
\newblock Technical Report RAL-TR-2008-008, Rutherford Appleton Laboratory,
  Oxfordshire, UK, 2008.

\bibitem{BS99}
A.~Bunse-Gerstner and R.~St{\"o}ver.
\newblock On a conjugate gradient-type method for solving complex symmetric
  linear systems.
\newblock {\em Linear Algebra and its Applications}, 287(1):105--123, 1999.

\bibitem{CJ91}
R.~H. Chan and X.-Q. Jin.
\newblock Circulant and skew-circulant preconditioners for skew-{H}ermitian
  type {T}oeplitz systems.
\newblock {\em BIT}, 31(4):632--646, 1991.

\bibitem{CJ07}
R.~H.-F. Chan and X.-Q. Jin.
\newblock {\em An Introduction to Iterative Toeplitz Solvers}.
\newblock Society for Industrial and Applied Mathematics, 2007.

\bibitem{CPP09}
X.-W. Chang, C.~C. Paige, and D.~Titley-P{\'e}loquin.
\newblock Stopping criteria for the iterative solution of linear least squares
  problems.
\newblock {\em SIAM J. Matrix Anal. Appl.}, 31(2):831--852, 2009.

\bibitem{C06}
S.-C.~T. Choi.
\newblock {\em Iterative Methods for Singular Linear Equations and
  Least-Squares Problems}.
\newblock PhD thesis, ICME, Stanford University, 2006.

\bibitem{C13c}
S.-C.~T. Choi.
\newblock {MINRES-QLP} pack and reliable reproducible research via staunch
  scientific software.
\newblock In {\em First Workshop on Sustainable Software for Science: Practice
  and Experiences}, Denver, Colorado, 2013.

\bibitem{CD02}
S.-C.~T. Choi, D.~L. Donoho, A.~G. Flesia, X.~Huo, O.~Levi, and D.~Shi.
\newblock About {B}eamlab---a toolbox for new multiscale methodologies.
\newblock \url{http://www-stat.stanford.edu/~beamlab/}, 2002.

\bibitem{CPS11}
S.-C.~T. Choi, C.~C. Paige, and M.~A. Saunders.
\newblock {MINRES-QLP}: {A} {K}rylov subspace method for indefinite or singular
  symmetric systems.
\newblock {\em SIAM J. Sci. Comput.}, 33(4):1810--1836, 2011.

\bibitem{MinresqlpMatlab}
S.-C.~T. Choi and M.~A. Saunders.
\newblock {MINRES-QLP} {MATLAB} package.
\newblock \url{http://code.google.com/p/minres-qlp/}, 2011.

\bibitem{CS12b}
S.-C.~T. Choi and M.~A. Saunders.
\newblock {ALGORITHM \& DOCUMENTATION: MINRES-QLP} for singular symmetric and
  {Hermitian} linear equations and least-squares problems.
\newblock Technical Report ANL/MCS-P3027-0812, Computation Institute,
  University of Chicago, IL, 2012.

\bibitem{MinresqlpF90}
S.-C.~T. Choi and M.~A. Saunders.
\newblock {MINRES-QLP} {FORTRAN 90} package.
\newblock \url{http://code.google.com/p/minres-qlp/}, 2012.

\bibitem{CS12}
S.-C.~T. Choi and M.~A. Saunders.
\newblock {ALGORITHM} xxx: {MINRES-QLP} for singular symmetric and {H}ermitian
  linear equations and least-squares problems.
\newblock {\em ACM Trans. Math. Software}, 2014.

\bibitem{C94}
J.~Claerbout.
\newblock Hypertext documents about reproducible research.
\newblock
  \url{http://sepwww.stanford.edu/doku.php?id=sep:research:reproducible}.

\bibitem{D09}
I.~S. Duff.
\newblock M{A}57---a code for the solution of sparse symmetric definite and
  indefinite systems.
\newblock {\em ACM Trans. Math. Software}, 30(2):118--144, 2004.

\bibitem{F76}
R.~Fletcher.
\newblock Conjugate gradient methods for indefinite systems.
\newblock In {\em Numerical Analysis (Proc 6th Biennial Dundee Conf., Univ.
  Dundee, Dundee, 1975)}, pages 73--89. Lecture Notes in Math., Vol. 506.
  Springer, Berlin, 1976.

\bibitem{FSJSU2}
L.~Foster.
\newblock {SJ}singular---{MATLAB} toolbox for managing the {SJSU} singular
  matrix collection.
\newblock \url{http://www.math.sjsu.edu/singular/matrices/SJsingular.html},
  2008.

\bibitem{FSJSU}
L.~Foster.
\newblock {S}an {J}ose {S}tate {U}niversity singular matrix database.
\newblock \url{http://www.math.sjsu.edu/singular/matrices/}, 2009.

\bibitem{F92}
R.~W. Freund.
\newblock Conjugate gradient-type methods for linear systems with complex
  symmetric coefficient matrices.
\newblock {\em SIAM J. Sci. Statist. Comput.}, 13(1):425--448, 1992.

\bibitem{FN91}
R.~W. Freund and N.~M. Nachtigal.
\newblock Q{MR}: a quasi-minimal residual method for non-{H}ermitian linear
  systems.
\newblock {\em Numer. Math.}, 60(3):315--339, 1991.

\bibitem{GL11}
D.~F. Gleich and L.-H. Lim.
\newblock Rank aggregation via nuclear norm minimization.
\newblock In {\em Proceedings of the 17th ACM SIGKDD International Conference
  on Knowledge Discovery and Data Mining}, KDD '11, pages 60--68, New York, NY,
  USA, 2011. ACM.

\bibitem{GV12}
G.~H. Golub and C.~F. Van~Loan.
\newblock {\em Matrix computations}.
\newblock Johns Hopkins University Press, 4th edition, 2012.

\bibitem{GV09}
C.~Greif and J.~Varah.
\newblock Iterative solution of skew-symmetric linear systems.
\newblock {\em SIAM J. Matrix Anal. Appl.}, 31(2):584--601, 2009.

\bibitem{Had1902}
J.~Hadamard.
\newblock {Sur les probl\`emes aux d\'eriv\'ees partielles et leur
  signification physique}.
\newblock {\em Princeton University Bulletin}, XIII(4):49--52, 1902.

\bibitem{HN96}
M.~Hanke and J.~G. Nagy.
\newblock Restoration of atmospherically blurred images by symmetric indefinite
  conjugate gradient techniques.
\newblock {\em Inverse Problems}, 12(2):157--173, 1996.

\bibitem{HO93}
P.~C. Hansen and D.~P. O'Leary.
\newblock The use of the {L}-curve in the regularization of discrete ill-posed
  problems.
\newblock {\em SIAM J. Sci. Comput.}, 14(6):1487--1503, 1993.

\bibitem{HS52}
M.~R. Hestenes and E.~Stiefel.
\newblock Methods of conjugate gradients for solving linear systems.
\newblock {\em J. Research Nat. Bur. Standards}, 49:409--436, 1952.

\bibitem{H08}
N.~J. Higham.
\newblock {\em Functions of matrices}.
\newblock Society for Industrial and Applied Mathematics, Philadelphia, PA,
  2008.
\newblock Theory and computation.

\bibitem{HJ}
R.~A. Horn and C.~R. Johnson.
\newblock {\em Matrix analysis}.
\newblock Cambridge University Press, Cambridge, 1990.
\newblock Corrected reprint of the 1985 original.

\bibitem{HJ91}
R.~A. Horn and C.~R. Johnson.
\newblock {\em Topics in matrix analysis}.
\newblock Cambridge University Press, Cambridge, 1991.

\bibitem{HC03}
D.~A. Huckaby and T.~F. Chan.
\newblock On the convergence of {S}tewart's {QLP} algorithm for approximating
  the {SVD}.
\newblock {\em Numer. Algorithms}, 32(2-4):287--316, 2003.

\bibitem{I84}
F.~Incertis.
\newblock A skew-symmetric formulation of the algebraic riccati equation
  problem.
\newblock {\em Automatic Control, IEEE Transactions on}, 29(5):467--470, 1984.

\bibitem{JLYY11}
X.~Jiang, L.-H. Lim, Y.~Yao, and Y.~Ye.
\newblock Statistical ranking and combinatorial {H}odge theory.
\newblock {\em Math. Program.}, 127(1, Ser. B):203--244, 2011.

\bibitem{JT10}
P.~Jir{\'a}nek and D.~Titley-P{\'e}loquin.
\newblock Estimating the backward error in {LSQR}.
\newblock {\em SIAM J. Matrix Anal. Appl.}, 31(4):2055--2074, 2010.

\bibitem{KS99}
M.~Kilmer and G.~W. Stewart.
\newblock Iterative regularization and {MINRES}.
\newblock {\em SIAM J. Matrix Anal. Appl.}, 21(2):613--628, 1999.

\bibitem{L56_88}
C.~Lanczos.
\newblock {\em Applied analysis}.
\newblock Englewood Cliffs, N.J., Prentice Hall, 1956.

\bibitem{P76}
C.~C. Paige.
\newblock Error analysis of the {L}anczos algorithm for tridiagonalizing a
  symmetric matrix.
\newblock {\em J. Inst. Math. Appl.}, 18(3):341--349, 1976.

\bibitem{PS75}
C.~C. Paige and M.~A. Saunders.
\newblock Solution of sparse indefinite systems of linear equations.
\newblock {\em SIAM J. Numer. Anal.}, 12(4):617--629, 1975.

\bibitem{PS82a}
C.~C. Paige and M.~A. Saunders.
\newblock {LSQR}: an algorithm for sparse linear equations and sparse least
  squares.
\newblock {\em ACM Trans. Math. Software}, 8(1):43--71, 1982.

\bibitem{PS82b}
C.~C. Paige and M.~A. Saunders.
\newblock \rlap{\phantom{z}}{Algorithm 583}; {LSQR}: Sparse linear equations
  and least-squares problems.
\newblock {\em ACM Trans. Math. Software}, 8(2):195--209, 1982.

\bibitem{SS86}
Y.~Saad and M.~H. Schultz.
\newblock G{MRES}: a generalized minimal residual algorithm for solving
  nonsymmetric linear systems.
\newblock {\em SIAM J. Sci. Statist. Comput.}, 7(3):856--869, 1986.

\bibitem{S97}
M.~A. Saunders.
\newblock Computing projections with {LSQR}.
\newblock {\em BIT}, 37(1):96--104, 1997.

\bibitem{SSY88}
M.~A. Saunders, H.~D. Simon, and E.~L. Yip.
\newblock Two conjugate-gradient-type methods for unsymmetric linear equations.
\newblock {\em SIAM J. Numer. Anal.}, 25(4):927--940, 1988.

\bibitem{SOL}
{S}ystems {O}ptimization {L}aboratory {(SOL)}, {S}tanford {U}niversity,
  downloadable software.
\newblock \url{http://www.stanford.edu/group/SOL/software.html}.

\bibitem{SV08}
P.~Sonneveld and M.~B. van Gijzen.
\newblock {IDR}$(s)$: A family of simple and fast algorithms for solving large
  nonsymmetric systems of linear equations.
\newblock {\em SIAM J. Sci. Comput.}, 31(2):1035--1062, 2008.

\bibitem{SS}
G.~Stewart and J.-G. Sun.
\newblock {\em Matrix perturbation theory}.
\newblock Computer science and scientific computing. Academic Press, 1990.

\bibitem{S69}
G.~W. Stewart.
\newblock On the continuity of the generalized inverse.
\newblock {\em SIAM J. Appl. Math.}, 17:33--45, 1969.

\bibitem{Stewart-Rice-1977}
G.~W. Stewart.
\newblock Research, development and {LINPACK}.
\newblock In J.~R. Rice, editor, {\em Mathematical Software III}, pages 1--14.
  Academic Press, New York, 1977.

\bibitem{S93}
G.~W. Stewart.
\newblock Updating a rank-revealing {$ULV$} decomposition.
\newblock {\em SIAM J. Matrix Anal. Appl.}, 14(2):494--499, 1993.

\bibitem{S99}
G.~W. Stewart.
\newblock The {QLP} approximation to the singular value decomposition.
\newblock {\em SIAM J. Sci. Comput.}, 20(4):1336--1348, 1999.

\bibitem{SW93}
D.~B. Szyld and O.~B. Widlund.
\newblock Variational analysis of some conjugate gradient methods.
\newblock {\em East-West J. Numer. Math.}, 1(1):51--74, 1993.

\bibitem{TB}
L.~N. Trefethen and D.~Bau, III.
\newblock {\em Numerical Linear Algebra}.
\newblock SIAM, Philadelphia, PA, 1997.

\bibitem{V92}
H.~A. Van~der Vorst.
\newblock Bi-{CGSTAB}: a fast and smoothly converging variant of {B}i-{CG} for
  the solution of nonsymmetric linear systems.
\newblock {\em SIAM J. Sci. Statist. Comput.}, 13(2):631--644, 1992.

\bibitem{W78}
O.~Widlund.
\newblock A {Lanczos} method for a class of nonsymmetric systems of linear
  equations.
\newblock {\em SIAM J. Numer. Anal.}, 15(4):801--812, 1978.

\end{thebibliography}
